\documentclass[10pt,journal,compsoc]{IEEEtran}
\usepackage{cite}
\ifCLASSINFOpdf
  \usepackage[pdftex]{graphicx}
\else
\fi

\usepackage{amsmath}
\usepackage{amssymb}
\usepackage{bm,bbm}
\usepackage{amsthm}
\usepackage{url}
\usepackage{xcolor}

\newtheorem{theorem}{Theorem}[section] 
\newtheorem{corollary}[theorem] {Corollary}

\newtheorem{lemma}[theorem]{Lemma}
\newtheorem{remark}{Remark}[theorem]

\begin{document}

\title{Coupling Asymmetry Optimizes Collective Dynamics over Multiplex Networks}

\author{Zhao~Song~and~Dane~Taylor
\IEEEcompsocitemizethanks{\IEEEcompsocthanksitem Z. Song is with the Department of Mathematics, Dartmouth College, Hanover, NH, 03755, USA and the Department of Mathematics, University at Buffalo, State University of New York, Buffalo, NY, 14260, USA. \protect\\
E-mail: zhao.song@dartmouth.edu
\IEEEcompsocthanksitem D. Taylor is with the School of Computing and the Department of Mathematics and Statistics, University of Wyoming, Laramie, WY, 82071, USA and the Department of Mathematics, University at Buffalo, State University of New York, Buffalo, NY, 14260, USA.\protect\\
E-mail: dane.taylor@uwyo.edu
}
\thanks{
This work was supported in part by the Simons Foundation (award \#578333) and the U.S. National Science Foundation (DMS-2052720).
}}

\markboth{Journal of \LaTeX\ Class Files,~Vol.~14, No.~8, August~2015}%
{Shell \MakeLowercase{\textit{et al.}}: Bare Demo of IEEEtran.cls for Computer Society Journals}

\IEEEtitleabstractindextext{
\begin{abstract} 
Networks are often interconnected, with one system wielding greater influence over another. However, the effects of such asymmetry on
self-organized phenomena (e.g., consensus and synchronization) are not well understood. Here, we study collective dynamics using a 
generalized graph Laplacian for multiplex networks containing layers that are asymmetrically coupled. We explore the nonlinear effects of coupling asymmetry on the convergence rate toward a  collective state, finding that asymmetry 
induces one or more optima that maximally accelerate convergence. When a faster and a slower system are coupled, depending on their relative timescales,  their optimal coupling is either \emph{cooperative}  (network layers mutually depend on one another) or \emph{non-cooperative} (one network directs another without a reciprocated influence). It is often optimal for the faster system to more-strongly influence the slower one, yet counter-intuitively,  the opposite can also be true.   As an application, we model collective 
decision-making for a human-AI system in which a social network is supported by  an AI-agent network, finding that a cooperative optimum requires that these two networks operate on a sufficiently similar timescale. More broadly, our work highlights the \emph{optimization of coupling asymmetry and timescale balancing} as fundamental concepts for the design of collective behavior over interconnected systems.  
\end{abstract}

}

\maketitle

\IEEEdisplaynontitleabstractindextext
\IEEEpeerreviewmaketitle
\IEEEraisesectionheading{\section{Introduction}\label{sec:introduction}}

\IEEEPARstart{C}{ollective} dynamics are widespread in nature and technology \cite{ermentrout2004review}  with applications ranging from synchronized oscillations in brains \cite{bansal2019cognitive} and power grids \cite{nishikawa2015comparative} to consensus processes in social networks \cite{hinsz1990cognitive,
fiol1994consensus,flood2000chief,
mohammed2001toward}, animal groups \cite{conradt2005consensus,westley2018collective}, and decentralized algorithms for machine learning and AI \cite{bijral2017data,assran2019stochastic,niwa2020edge,vogels2020powergossip,kong2021consensus,huynh2021}. 
The formulation of many such models involves a graph  Laplacian matrix $\mathbf{L}$ whose entries encode a network,  including models for  the synchronization of networks of identical dynamical systems \cite{pecora1998master,sun2009master}  and heterogeneous phase oscillators \cite{skardal2014optimal,taylor2016synchronization}, consensus dynamics \cite{kibangou2012graph,huynh2021},    Markov chains \cite{delvenne2010stability,boyd2004fastest},  diffusion \cite{masuda2017,coifman2006diffusion}, and DC electricity flow \cite{doyle1984random}. The analyses of such systems often utilize spectral theory, and in particular, the second-smallest eigenvalue $\lambda_2$ of $\mathbf{L}$ can help determine dynamical properties such as convergence rate \cite{kibangou2012graph,huynh2021} and local stability  \cite{pecora1998master,sun2009master}. 

One should note, however, that systems rarely exist in total isolation, and it is important to understand the dynamics of interconnected (i.e., multilayer) networks \cite{kivela2014,boccaletti2014}. A popular modeling framework is multiplex networks \cite{cozzo2018multiplex}, consisting of \emph{network layers}, each involving the same set of nodes but possibly different edges (called \emph{intralayer edges}). In recent years, there has been growing interest in extending Laplacian-based  models  to the setting of multiplex networks, including work on random walks   \cite{mucha2010,sole2016random,ding2018centrality,
taylor2020multiplex}, 
 synchronization \cite{aguirre2014synchronization,gambuzza2015intra,jalan2016cluster,sawicki2018delay,liu2020intralayer},
 and diffusion \cite{gomez2013diffusion,sole2013spectral,deford2018new}.
Of particular importance are  \emph{supraLaplacian matrices} \cite{gomez2013diffusion,sole2013spectral} that generalize graph Laplacians to  multiplex networks, thereby extending the general field of Laplacian-based dynamics to this setting. 
Importantly, existing research on supraLaplacian matrices and related applications has focused on multiplex networks in which the layers are symmetrically coupled using undirected \emph{interlayer edges}. This is problematic since the effects of asymmetry  are known to  play a crucial role in shaping  self-organization for network-coupled dynamical systems \cite{bragard2003asymmetric,timme2006does,restrepo2014mean,skardal2015erosion,nishikawa2016symmetric,asllani2018structure,molnar2021asymmetry}, and they provide opportunities for system optimization \cite{taylor2020introduction}. Moreover, it is natural to assume that the relationship between networks is asymmetric for many contexts. 

Here, we propose and analyze a model for interconnected consensus  systems, which can represent, for example, 
collective decision-making over a social network that  is supported by AI agents, which provide decision support and themselves interact and cooperatively learn. (See Sec.~\ref{sec:AI} for further description.) Human-AI  systems are gaining popularity for decision-making in military \cite{rasch2003incorporating,zhou2019bayesian,azar2021drone} 
and financial contexts \cite{albadvi2007decision,yu2005designing,chou1997stock}, 
yet existing theory  for interconnected decision systems is insufficient (even, as we shall show, for a simple linear model). 
For this application, a dystopian-minded engineer would naturally design the social network   to wield greater influence over the network  of AI agents, \emph{but how can this be achieved? And how might such a  system be optimized?} Similar questions arise for any collective dynamics over asymmetrically coupled networks, as well as the following question: \emph{When networks are optimally coupled, is their coupling configuration cooperative or non-cooperative?} That is, do optimally coupled networks mutually influence each other, or does one network direct others without  feedback.

In this work, we approach these questions by considering networks that are optimally coupled to maximize the convergence rate toward a collective state. We propose a generalized {supraLaplacian matrix} \cite{tejedor2018diffusion,sole2013spectral,wang2021unique,gomez2013diffusion,taylor2020multiplex}  
$\mathbb{L}(\omega,\delta)$, where $\omega\ge0$ is a \emph{coupling strength} that controls how strongly network layers influence each other and 
$\delta\in[-1,1]$ is an \emph{asymmetry parameter} that tunes the extent to which interlayer edges are biased in a particular direction.
Motivated by human-AI decision systems, we use $\mathbb{L}(\omega,\delta)$ to formulate a continuous-time linear model for interconnected consensus systems. 
We find that coupling asymmetry can significantly bias the limiting consensus state and possibly speed or slow the  rate $\text{Re}(\lambda_2)$ of convergence (which depends on  the real part of $\lambda_2$, since $\mathbb{L}(\omega,\delta)$ can be  an asymmetric matrix).

To gain analytical insight,  we develop spectral approximation theory for  the large $\omega$ (i.e., strong coupling) limit to identify and characterize different effects on $\text{Re}(\lambda_2)$ due to varying  $\delta$. This reveals several surprising and unintuitive insights. For example, depending on the network layers' structures, increasing the  magnitude (i.e., $|\delta|$)  of  coupling asymmetry can  monotonically slow  convergence,  monotonically speed convergence, or have a more complicated effect on $\text{Re}(\lambda_2)$. For some systems, the direction (i.e., $\text{sign}(\delta)$) of asymmetry is extremely important, whereas it doesn't matter for others. We provide an initial identification and taxonomy for such nonlinear behaviors.

Because technological and natural  systems are often highly optimized  due to engineering and the process of natural selection, we present theory and experiments to study network layers that are coupled with an optimal level of asymmetry,  $\hat{\delta}=\text{argmax}_\delta \text{Re}(\lambda_2)$, that maximally accelerates convergence toward a collective state.
Focusing on the case of two layers, we characterize these configurations as being either \emph{cooperative}, in which case $|\hat{\delta}|<1$ so that $\delta$ lies within the open set $(-1,1)$ and the   layers  mutually   influence each other;   or \emph{non-cooperative}, in which case    $|\hat{\delta}|=1$ so that $\hat{\delta}\in\{-1,1\}$ lies on the boundary. In the latter case,  the optimal asymmetry involves one network   fully directing the      other   without feedback. Notably, the existence of a cooperative optimum  guarantees that the convergence rate of the multiplexed systems is faster than that for either system.

We find that the nonlinear effects of $\delta$ on $\text{Re}(\lambda_2)$ and   optima $\hat{\delta}$   depend sensitively on the layers' distinct  topological structures as well as their separate  time scales for consensus. 
Therefore, we introduce and   study a \emph{rate-scaling} parameter $\chi$ 
that allows us to tune  whether consensus is faster within layer 1 ($\chi\approx1$) or  layer 2  ($\chi\approx0$).
By considering the range $\delta\in[-1,1]$ for fixed $\chi \in(0,1)$, we obtain a  criterion (see Sec.~\ref{sec:exist}) that can guarantee the existence of a cooperative optimum, which requires that the layers' dynamics have sufficiently similar time scales (i.e., $\chi $ is neither too large or small). Finally, we also consider   situations where both $\delta$ and $\chi$ can be   freely varied  and  jointly optimized. We also  identify scenarios of cooperative and non-cooperative optima for this more complicated setting, showing that it can be beneficial to design one  layer to be as fast as possible and then have that layer non-cooperatively influence other layers without feedback. However, for other network structures, convergence can be fastest by striking a cooperative balance, both in terms of the asymmetric coupling of layers as well as a balance for their respective time scales.
Our work  highlights \emph{optimization through coupling asymmetry and time-scale balancing} as   important directions for understanding and engineering collective dynamics over   human-AI consensus systems and other interconnected networks in general. Moreover, because graph Laplacian matrices are so widely used to study physical, biological and technological systems, our findings are relevant and broadly informative for many other types of dynamics.

This paper is organized as follows. In Sec.~\ref{sec:model}, we introduce the model that we study. In Sec.~\ref{sec:asymmetry}, we present experiments highlighting various effects of coupling asymmetry. In Sec.~\ref{sec:theory}, we present theoretical results including an existence guarantee for   a  cooperative optimum. In Sec.~\ref{sec:AI}, we apply the framework to  model  collective decisions by  human-AI teams.
A discussion is provided in Sec.~\ref{sec:Disc}.

\vspace{-.2cm}

\section{Model}\label{sec:model}

We first define a model    for multiplex networks   with asymmetrically coupled   layers (Sec.~\ref{sec:model1}) and a model for collective dynamics over such networks (Sec.~\ref{sec:model2}).
Our formulation has three tunable parameters: $\omega$ and $\delta$ control the strength and asymmetry of coupling between layers, respectively, whereas $\chi$ controls the different timescales of dynamics in separate layers.

\subsection{Multiplex Networks with Asymmetric Coupling}\label{sec:model1}

We begin by defining supraLaplacian matrices for multiplex networks with asymmetrically coupled   layers and by formulating an interconnected consensus model for collective dynamics. Consider a multiplex network with $T$ network layers, each consisting of $N$ nodes. For each   layer $t\in\{1,\dots,T\}$, we let ${\bf A}^{(t)}\in\mathbb{R}^{N\times N}$ be its ``{intralayer}'' adjacency matrix and $\mathbf{L}^{(t)}={\bf D}^{(t)}-{\bf A}^{(t)}$ be its {intralayer}  unnormalized Laplacian matrix, where  ${\bf D}^{(t)}$ is a diagonal matrix that encodes the nodes'  weighted in-degrees, $\mathrm{D}_{ii}^{(t)}=\sum_j \mathrm{A}_{ij}^{(t)}$ (also called `strengths'). Note that our notational convention is to let $A_{ij}>0$ encode the weight for an edge from node $j$ to $i$.
Matrices ${\bf A}^{(t)}$ and $\mathbf{L}^{(t)}$ are size $N\times N$ and are asymmetric if  network  layer $t$  contains directed edges. 

We couple the layers using an  ``interlayer'' adjacency matrix 
$\bm{A}^\mathrm{I}(\delta) =  (1+\delta) \hat{{\bm{A}}}^\mathrm{I} + (1-\delta)[ \hat{{\bm{A}}}^\mathrm{I}]^{\mathrm{T}}$,  
where $\delta\in[-1,1]$ is an \emph{asymmetry parameter} that tunes  the magnitude and direction of coupling asymmetry and  $\hat{{\bm{A}}}^\mathrm{I}$ is an adjacency matrix for a graph in which all edges are strictly directed (i.e., there no bidirectional edges or self edges).   Each positive matrix element $\hat{A}_{st}^\mathrm{I}>0$ encodes a directed influence from network layer $t$ to layer $s$.
When $\delta\not=0$, the coupling between layers is biased in a particular direction that is encoded by $\hat{{\bm{A}}}^\mathrm{I}$, whereas the coupling is symmetric when $\delta=0$.
We refer the situation of $\delta\in\{-1,1\}$ as \emph{fully asymmetric coupling}, and it can possibly yield situations in which a  network layer influences other layers but itself is not influenced by any other layer. Note also that we can equivalently define
$ \bm{A}^\mathrm{I}(\delta)=  {\bm{A}}^{\mathrm{I}}_+ +\delta  {\bm{A}}^{\mathrm{I}}_-$, where $  {\bm{A}}^\mathrm{I}_+ = \hat{\bm{A}}^\mathrm{I}+ [\hat{\bm{A}}^\mathrm{I}]^{\mathrm{T}}$ and $  {\bm{A}}^\mathrm{I}_- = \hat{\bm{A}}^\mathrm{I}- [\hat{\bm{A}}^\mathrm{I}]^{\mathrm{T}}$ are symmetric and skew-symmetric matrices, respectively.

Given $\bm{A}^\mathrm{I}(\delta) $, we define an associated interlayer unnormalized Laplacian  $\bm{L}^\mathrm{I}(\delta)=\bm{D}^\mathrm{I}(\delta)-\bm{A}^\mathrm{I}(\delta)$, where  $\bm{D}^\mathrm{I}(\delta)$ is a diagonal matrix with entries 
${D}^\mathrm{I}_{ss} (\delta)= \sum_t A_{st}^\mathrm{I}(\delta)$.
%
It is also useful to define an equivalent formulation,
\begin{equation}
\bm{L}^\mathrm{I}(\delta)=  {\bm{L}}^\mathrm{I}_+ +\delta  {\bm{L}}_-^\mathrm{I},
\label{eq:Ld}
\end{equation}
where 
$ {\bm{L}}^\mathrm{I}_+ =  {\bm{D}}^\mathrm{I}_+ -  {\bm{A}}^\mathrm{I}_+$, $ {\bm{L}}^\mathrm{I}_- =  {\bm{D}}^\mathrm{I}_-  - {\bm{A}}^\mathrm{I}_-$,  $[D^\mathrm{I}_+]_{ss} = \sum_t  [A^\mathrm{I}_+]_{st} $,   and $ [D^\mathrm{I}_-]_{ss} = \sum_t [A^\mathrm{I}_-]_{st} $. Note that $ {\bm{L}}^\mathrm{I}_+$ is a Laplacian matrix that is associated with an undirected graph, whereas $ {\bm{L}}^\mathrm{I}_-$ may be interpreted as a Laplacian for a directed, signed graph that has a very particular structure:  for   any positive edge weight $[A^\mathrm{I}_-]_{st} >0$, the reciprocal edge must exist and have negative  weight $[A^\mathrm{I}_-]_{ts}  = -[A^\mathrm{I}_-]_{st} $. (We note that there are   other ways to define   signed Laplacians \cite{kunegis2010spectral,pan2016laplacian}.)
In the case of  $T=2$   layers, such as the multiplex network shown in Fig.~\ref{fig1}, we define   $\hat{\bm{A}}^\mathrm{I}  = \begin{pmatrix} 0 & 0\\ 1& 0 \end{pmatrix}$ and
\begin{align}\label{eq:LI_2}
    \bm{L}^{ \mathrm{I}} (\delta)
	&=
	\begin{pmatrix}
	1 & -1\\
	-1& 1
	\end{pmatrix}
	+
		\delta\begin{pmatrix}
	 -1& 1\\
	-1& 1
	\end{pmatrix}.
\end{align}

We next define a supraLaplacian  matrix following \cite{gomez2013diffusion}
by scaling each $\bm{L}^\mathrm{I}(\delta)$ by a \emph{coupling strength} $\omega>0$ to construct a supraLaplacian matrix 
\begin{equation}\label{eq:supra}
    \mathbb{L}(\omega,\delta) = \mathbb{L}^{\mathrm{L}} + \omega\mathbb{L}^{\mathrm{I}} (\delta),   
\end{equation}
where $\mathbb{L}^{\mathrm{L}} = \text{diag}[  \mathbf{L}^{(1)} ,\dots ,\mathbf{L}^{(T)}]$ contains intralayer Laplacians as diagonal blocks, and $\mathbb{L}^{\mathrm{I}} (\delta)= \bm{L}^\mathrm{I}(\delta)\otimes {\bf I}$ couples the layers in a way that is \emph{uniform} (i.e., any   coupling between two given layers is the same) and \emph{diagonal} (i.e., any coupling between layers connects a node in one layer to itself in another layer) \cite{taylor2017eigenvector,taylor2019tunable}.
Symbol $\otimes$ indicates the Kronecker product. Note that $\mathbb{L}(\omega,\delta)$ is a size-$(NT)$ square matrix, and we will enumerate its rows an columns by $p\in\{1,\dots,NT\}$.
 
Under the choice $\delta=0$, $\mathbb{L}(\omega,\delta)$ is a symmetric matrix, and it recovers previously studied supraLaplacians \cite{tejedor2018diffusion,sole2013spectral,wang2021unique,gomez2013diffusion,taylor2020multiplex}, which
that  have been used to study diffusion and random walks over multiplex networks with layers that are symmetrically coupled using undirected interlayer edges. Understanding the spectral properties of supraLaplacians has revealed   novel insights including   ``superdiffusion''   \cite{tejedor2018diffusion,cencetti2019diffusive,wang2021unique}, 
whereby diffusion over coupled networks is faster than that of any single network layer, if it were in isolation. 
The study of synchronization over multiplex networks has similarly led to discoveries including the observation the networks coupled with moderate coupling strength have better synchronizability \cite{sole2013spectral} and other insights \cite{aguirre2014synchronization,gambuzza2015intra,jalan2016cluster,sawicki2018delay,liu2020intralayer}. Despite this progress, the effects of asymmetric coupling  on multiplex-network dynamics remains under-explored.

\begin{figure}[t]
	\centering 
	\includegraphics[width=3.1in]{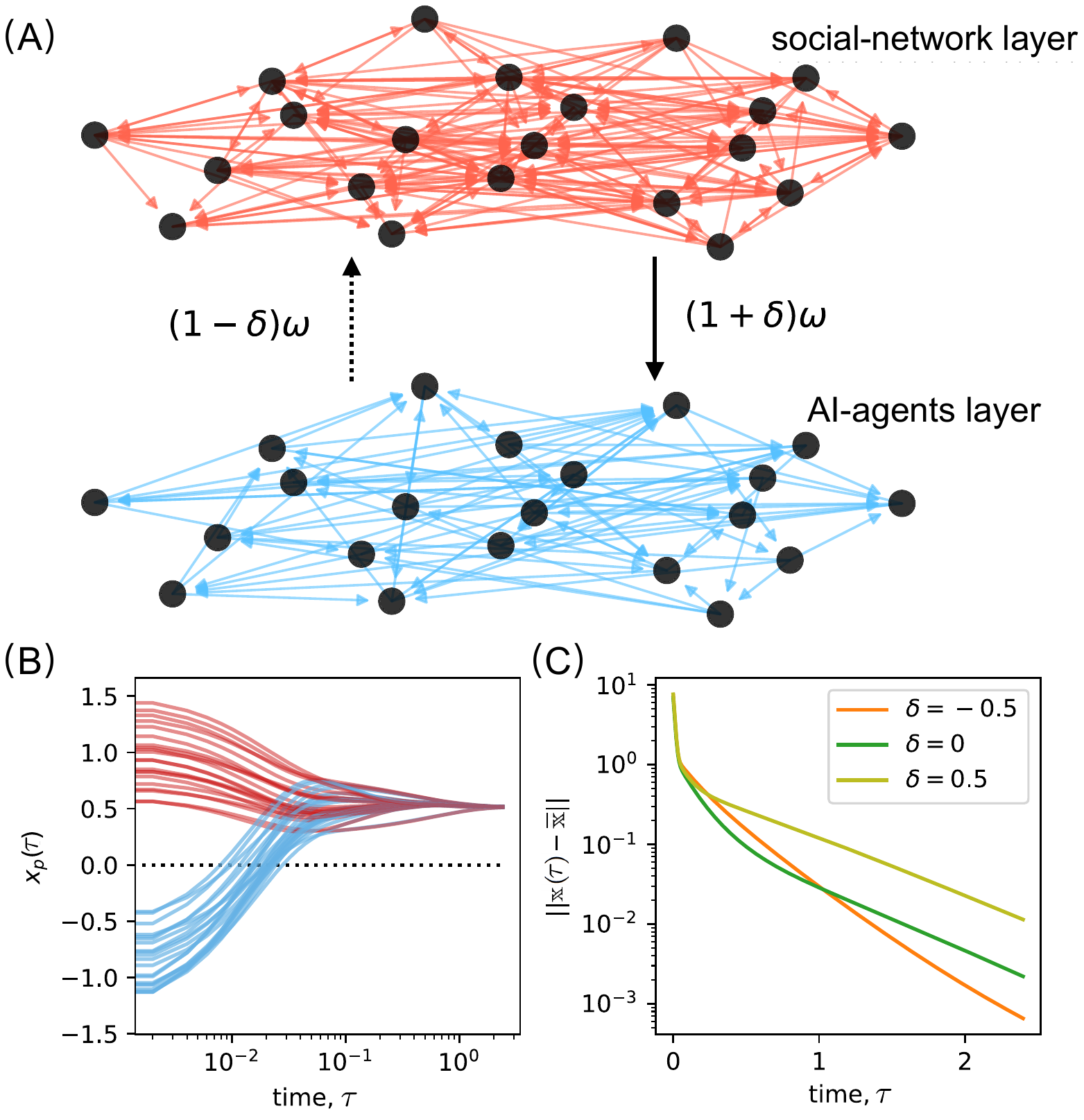}
	\vspace{-.2cm}
	\caption{
		{\bf Asymmetric coupling  biases the  limit and convergence rate for interconnected consensus systems.}
		(A)~A two-layer network in which layer 1 is an empirical social network that encodes mentoring relationships among corporate executives \cite{krackhardt1987cognitive}. Layer 2 models a communication network over which AI-agents cooperatively learn, which we generate as a random directed graph. The AI-agents provide decision support to the executives through   ``interlayer'' edges of weight $(1 \pm \delta)\omega$, where $\omega$ and $\delta$ encode the  strength and asymmetry, respectively, of coupling between the two networks.
		(B)~Nodes' states $x_p(\tau)$  reach consensus (i.e., a collective decision) following \eqref{eq:con}, which uses a supraLaplacian matrix $\mathbb{L}(\omega,\delta)$ with $(\omega,\delta)=(30,0.5)$ and converges to a limit $\overline{x}$ that is a weighted average of the initial condition with weights that are biased by $\delta$.
		(C)~Convergence  $x_p(\tau) \to\overline{x}$ is shown to occur faster  for $\delta=-0.5$ than for $\delta \in\{0,0.5\}$. For each $\delta$, the   convergence rate $\text{Re}(\lambda_2)$ gives the curve's asymptotic slope for large $\tau$.
	}
	\label{fig1}
\end{figure}

\subsection{Asymmetrically Coupled Consensus Systems}\label{sec:model2}

Consensus  is a popular model for collective decision  making  in the cognitive, social and biological sciences \cite{hinsz1990cognitive, 
flood2000chief,mohammed2001toward,conradt2005consensus}, and it  also provides a foundation for  decentralized algorithms for neural networks and machine learning \cite{bijral2017data,assran2019stochastic,niwa2020edge,vogels2020powergossip,huynh2021,kong2021consensus}. 
Thus motivated, we propose a model for interconnected consensus systems via the following linear ordinary differential equation,
\begin{align}\label{eq:con}
    \frac{d}{d\tau} \mathbbm{ x}(\tau) &= - \mathbb{L}(\omega,\delta) \mathbbm{ x}(\tau),
\end{align}
where $\mathbbm{x}(\tau)=[x_1(\tau),\dots, x_p(\tau),\dots, x_{NT}(\tau)]^\mathrm{T}$ is a length-$NT$ vector. 
Each $x_p(\tau)$ encodes the state of node $i_p = (p \text{ mod } N)$ in layer $t_p = \lceil p/N \rceil $  at time $\tau$.  (We  let  $p\in\{1,\dots,NT\}$ and use mod$(\cdot)$ and $\lceil \cdot \rceil $ to denote the modulus and ceiling function, respectively.) 
In Sec.~\ref{sec:AI}, we interpret \eqref{eq:con} as a simple-yet-informative model for collective decisions in a human-AI system in which individuals in a social network  are supported by AI agents, who themselves coordinate and collectively learn.

Equation \eqref{eq:con} can be considered as the ``multiplexing'' of two consensus processes: consensus within each network layer and consensus across layers. That is, one could define an intralayer consensus dynamics  for each network layer $t$:
$\frac{d}{d\tau} {\bf x}^{(t)}(\tau) = - \mathbf{L}^{(t)}  {\bf x}^{(t)} (\tau)$ with initial condition ${\bf x}(0)\in\mathbb{R}^N$. 
Similarly, one can define an interlayer consensus dynamics by
$\frac{d}{d\tau} \bm{ x}(\tau) = - \omega\bm{L}^\mathrm{I}(\delta) \bm{ x} (\tau)$ with some initial condition $\bm{x}(0)\in\mathbb{R}^{T}$. 
In this context, the scaling by $\omega$ controls the   timescale of consensus   across layers as compared to consensus within layers. As such, it is important to consider dynamics for a wide range of $\omega$ values. 
We also note that each of these differential equations can be interpreted as a type of Abelson model \cite{abelson1964mathematical,abelson1967mathematical} for opinion dynamics.


One can also vary the  timescales for dynamics and consensus within each separate layer by scaling each $\mathbf{L}^{(t)} $ by some nonnegative constant. 
Focusing on the case of $T=2$ layers, replace the intralayer Laplacians by 
   $ \mathbf{L}^{(1)}\mapsto \chi \mathbf{L}^{(1)}$ and
    $\mathbf{L}^{(2)} \mapsto (1-\chi) \mathbf{L}^{(2)}$
where 
$\chi\in(0,1)$ is a a \emph{rate-scaling parameter} that
controls whether  the layers' timescales are equally balanced ($\chi\approx 0.5$), whether layer 1 is much faster than layer 2 ($\chi\approx1$), or vice versa ($\chi\approx0$). 
We will initially not include $\chi$ in our model and will investigate it later in Sec.~\ref{sec:exist} and Sec.~\ref{sec:AI}.


Although we focus here on consensus dynamics, it would be straightforward to utilize matrix
$\mathbb{L}(\omega,\delta)$ to formulate models for  diffusion, synchronization and other Laplacian-based dynamical processes over multiplex networks with asymmetrically coupled layers. For example, the substitution of $\mathbb{L}(\omega,\delta)\mapsto\mathbb{L}(\omega,\delta)^{\mathrm{T}}$ in \eqref{eq:con} would yield a model for diffusion.

\begin{figure*}[!t]
	\centering
	\includegraphics[width=1\linewidth]{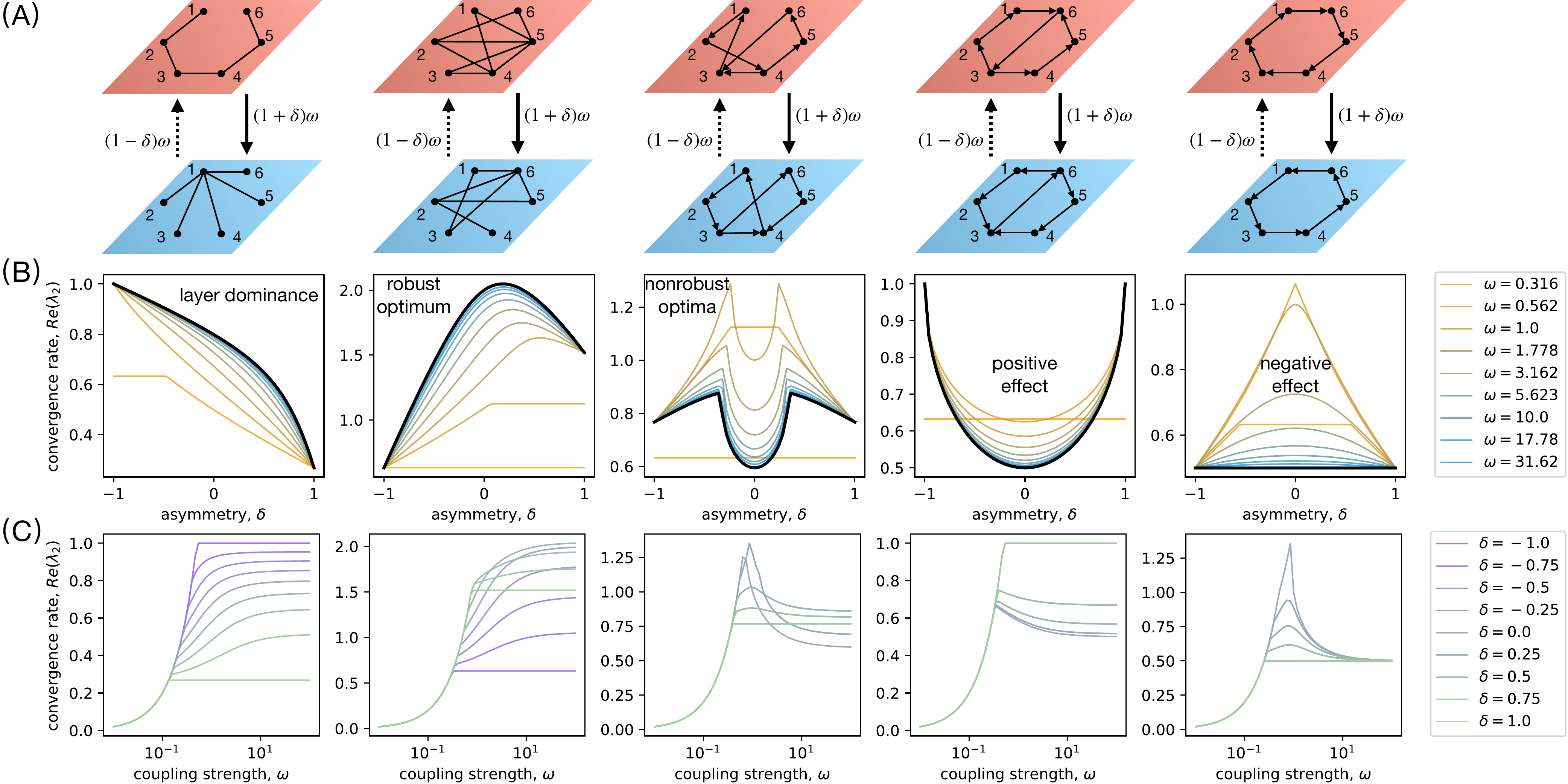}
	\caption{{\bf  Characteristic ways in which coupling asymmetry can affect convergence rate.}
		(A)~Interconnected consensus systems 
		with $N=6$ nodes and $T=2$ layers that are asymmetrically coupled through the interlayer Laplacian matrix defined in \eqref{eq:LI_2}.
		(B)~For each respective system, we plot the convergence rate $\text{Re}(\lambda_2)$  versus $\delta$. Colored curves depict different   $\omega$ (see legend). Observe for different networks that  asymmetry (i.e., $\delta\neq 0$) has remarkably diverse effects in which it either speeds/slows convergence depending on the sign and magnitude of $\delta$ (e.g., see columns 1 and 2), always accelerates  convergence (e.g., see columns 3 and 4),  or always slows convergence (e.g., see column  5). Black curves depict our theoretical prediction for large $\omega$ given by \eqref{eq:lim_lam2}--\eqref{eq:L_bar_2}.
		(C)~For each  system, we   plot $\text{Re}(\lambda_2)$ versus coupling strength $\omega$.  Observe for all systems that $\delta$ impacts $\text{Re}(\lambda_2)$ only when $\omega$ is sufficiently large and that some systems exhibit an \emph{intermediate coupling optimum} (ICO), which is a peak that can either be exaggerated or diminished by introducing coupling asymmetry.
		}
	\label{fig:behaviors_toy}
\end{figure*}

\section{  Effects of Coupling Asymmetry}\label{sec:asymmetry}

We first show how coupling asymmetry can have various effects on solutions to 
\eqref{eq:con}. In Sec.~\ref{sec:limit_rate}, we highlight that asymmetry can   bias the consensus limit and either increase or decrease the convergence rate.
In Sec.~\ref{sec:asaaq}, we discuss the optimization of coupling asymmetry $\delta$ to maximize the convergence rate.
 In Sec.~\ref{sec:large}, we study the effects on  random multiplex networks.

\subsection{Impact on Consensus Limit and Convergence Rate}\label{sec:limit_rate}

We assume that the smallest eigenvalue $\lambda_1 =0$ of $\mathbb{L}(\omega,\delta)$ is simple (i.e., has algebraic and geometric multiplicities equal to one), which is guaranteed, e.g., if the interlayer network and aggregated network are both strongly connected \cite{taylor2020multiplex}. Then  for any real-valued initial condition $\mathbbm{  x}(0) $,  \eqref{eq:con} converges to an equilibrium $\mathbbm{x}( {\tau} )\to \overline{\mathbbm{ x}}= [\overline{x},\dots, \overline{x}]^{\mathrm{T}}$, which is the right eigenvector associated with $\lambda_1 $. The \emph{consensus limit} is reached at a scalar value  
\begin{align}
    \overline{x} = \sum_p u_p x_p(0) / \sum_{p'}u_{p'},
\end{align}
which is a weighted average of the initial states, and the weights $u_p$ are entries of the left eigenvector $\mathbbm{u}$ of $\mathbb{L}(\omega,\delta)$ that is associated with $\lambda_1$. The \emph{asymptotic convergence rate},       
\begin{align}\label{eq:conv}
       -\limsup_{\mathbbm{  x}(0) } \frac{1}{\tau^*} \int_{0}^\mathrm{\tau^*}  \log \left(\frac{||\mathbbm{  x}(\tau) - \overline{\mathbbm{x}}||}{||\mathbbm{x}(0) - \overline{\mathbbm{  x}}||} \right) d\tau \le \text{Re}(\lambda_2),
\end{align}
is bounded by $\text{Re}(\lambda_2)$, the real part of the eigenvalue of $\mathbb{L}(\omega,\delta)$ that has   second-smallest real part. Throughout this manuscript, we refer to $ \text{Re}(\lambda_2)$ simply as the ``convergence rate'', keeping in mind  that it is a bound on the asymptotic behavior of convergence.

In Fig.~\ref{fig1}, we show how coupling asymmetry can bias the collective state $\overline{x}$  as well as  the   convergence rate $\text{Re}(\lambda_2)$.  
In this example,  we study  \eqref{eq:con} for $T=2$ network   layers with $\bm{L}^{ \mathrm{I}} (\delta)$ given by
\eqref{eq:LI_2}. The first intralayer consensus system   models  group decision making within an empirical social network that encodes mentoring relationships among corporate executives  \cite{krackhardt1987cognitive}. It  contains $N=21$ nodes and    190 directed edges and was downloaded from \cite{man_data}. 
The second intralayer consensus system represents  AI agents that support the executives' decisions, and we model their communication by a directed $4$-regular graph  that we generated using the configuration model. We note that we have selected this empirical example because  collective human-AI decision making is   already widespread in military \cite{rasch2003incorporating,zhou2019bayesian,azar2021drone} 
and financial contexts \cite{albadvi2007decision,chou1997stock,yu2005designing}, 
and we predict that it will become increasingly popular in corporate settings in the near future.

In Fig.~\ref{fig1}(B), we show  converging trajectories $x_p( {\tau} )$ under the parameter choices $(\omega,\delta)=(30,0.5)$   
for an initial condition where $x_p( {\tau} )$ are positive for nodes in layer 1 and negative for nodes in layer 2.  Observe  that  the limit $\overline{x}>0$ is biased to be positive, implying that the social network has a ``stronger say''  than the AI agents regarding the   state at which consensus is reached. This occurs here because we chose $\delta>0$,  implying that the social network more strongly influences the AI agents, that is, as compared to the reciprocal relationship.
However, observe in Fig.~\ref{fig1}(C)  that  convergence is slower for this value of $\delta$ as compared to the other shown values: $\delta\in\{0,-0.5\}$. 
Later in Sec.~\ref{sec:AI}, we will present an extended study of the network shown in Fig.~\ref{fig1}, showing that $\text{Re}(\lambda_2)$ has a maximum near  $ {\delta}\approx -0.6$. We will also further discuss the implications of our work for the application area of human-AI systems.

\subsection{Optimizing Asymmetry for Fast Convergence}\label{sec:asaaq}

For many applications it is beneficial for consensus to converge quickly, and we are particularly interested in understanding the value (or values) of $\delta$ that maximize the convergence rate:
$\hat{\delta}=\text{argmax}_\delta \text{Re}(\lambda_2)$. To this end, it is beneficial to first gain a broader understanding for the diverse ways that $\delta$ can influence $\text{Re}(\lambda_2)$. Below, we
offer a characterization of several possible ways for how the convergence rate $\text{Re}(\lambda_2)$ toward a collective state can be affected by the asymmetric coupling of  network layers. The columns of Fig.~\ref{fig:behaviors_toy}  highlight  five  distinct behaviors for how asymmetric coupling can affect $\text{Re}(\lambda_2)$ and its optima  $\hat{\delta}=\text{argmax}_\delta \text{Re}(\lambda_2)$. These are:

\begin{itemize}
    \item[(i)] \emph{layer dominance:} $\text{Re}(\lambda_2)$ monotonically increases or decreases with $\delta$ and obtains a maximum  at either $\hat{\delta}   =  \pm1$; 
    \item[(ii)] \emph{robust cooperative optimum:} $\text{Re}(\lambda_2)$ obtains a  maximum at   some value $\hat{\delta}\in(-1,1)$, and $\text{Re}(\lambda_2)$ is differentiable with respect to $\delta$ at the optimum;
    \item[(iii)] \emph{nonrobust cooperative optima:} $\text{Re}(\lambda_2)$ obtains a maximum at some  value  $\hat{\delta}\in(-1,1)$, but $\text{Re}(\lambda_2)$ is not differentiable with respect to $\delta$ at the optimum;
    \item[(iv)] \emph{strictly positive effect:} $\text{Re}(\lambda_2)$ is a non-decreasing function of  $|\delta|$ (i.e., slowest  convergence    at $\delta=0$);
    \item[(v)] \emph{strictly negative effect:} $\text{Re}(\lambda_2)$ is a non-increasing function of $|\delta|$ (i.e., fastest  convergence    at $ {\delta}=0$).
\end{itemize} 
 
Notably, these characterizations are an incomplete list. Future work will likely reveal other possible behaviors, thereby broadening our understanding of how the asymmetric coupling of networks can impact a combined system's convergence rate and other dynamical properties for collective dynamics. 

We provide    examples that  exhibit the behaviors (i)--(v) in the five columns, respectively,  of  Fig.~\ref{fig:behaviors_toy}. We visualize the multiplex networks in Fig.~\ref{fig:behaviors_toy}(A), and for each system, we plot $\text{Re}(\lambda_2)$ versus $\delta$  in Fig.~\ref{fig:behaviors_toy}(B) for different choices of $\omega$. The solid black curves in Fig.~\ref{fig:behaviors_toy}(B) are a theoretical predict that we will present   in Sec.~\ref{sec:thy}. The systems depicted in columns 1 and 4 of Fig.~\ref{fig:behaviors_toy} exhibit behaviors (i) and (iv), respectively, which exhibit non-cooperative optima.  That is, convergence toward the collective state occurs optimally fast only when one   layer influences the other   without a reciprocated influence (i.e., $\hat{\delta}\pm1$). In contrast, the systems depicted in columns 2, 3 and 5   exhibit behaviors (ii), (iii) and (v),   which have cooperative  optima. Convergence is fastest when the layers mutually influence one another (i.e., $|\hat{\delta}|<1$).
We further note that  it is beneficial to distinguish between robust vs non-robust optima, and we   call an optimum ``robust'' if and only if   $\frac{d}{d\delta}\text{Re}(\lambda_2) $ is zero  at the optimum. Robustness is important when considering the possible effects on $\text{Re}(\lambda_2)$ of  uncertainty for $\delta$ near the optimum.  

In Fig.~\ref{fig:behaviors_toy}(C), we  plot $ \text{Re}(\lambda_2)$  versus   $\omega$ to  investigate the relation between coupling asymmetry $\delta$ and coupling strength $\omega$. First, observe that there exists a peak for the third, fourth and fifth systems for some intermediate value of $\omega$, but not the  first or second system. Such a peak corresponds to an optimal choice of $\omega$ at which the convergence rate maximized.
We refer to this phenomenon as an \emph{intermediate coupling optimum} (ICO), and this spectral property for $\lambda_2$ has been previously called ``superdiffusion''  \cite{tejedor2018diffusion,cencetti2019diffusive,wang2021unique}
 in the context of diffusion on multiplex networks. In fact,  the first, second and fifth  systems were previously used to study diffusion over multiplex networks  in \cite{taylor2020multiplex},  \cite{gomez2013diffusion} and \cite{tejedor2018diffusion}, respectively. Those studies were restricted to the assumption of symmetric coupling in which  $\delta=0$. Extending that work, here we show that there exists a similar ICO phenomenon for multiplexed consensus systems, and that the size of the peak can be either exaggerated or inhibited depending on the direction and magnitude of coupling asymmetry.

\subsection{Effects on Random Multiplex Networks}\label{sec:large}

\begin{figure*}[h]
	\centering
	\includegraphics[width=1\linewidth]{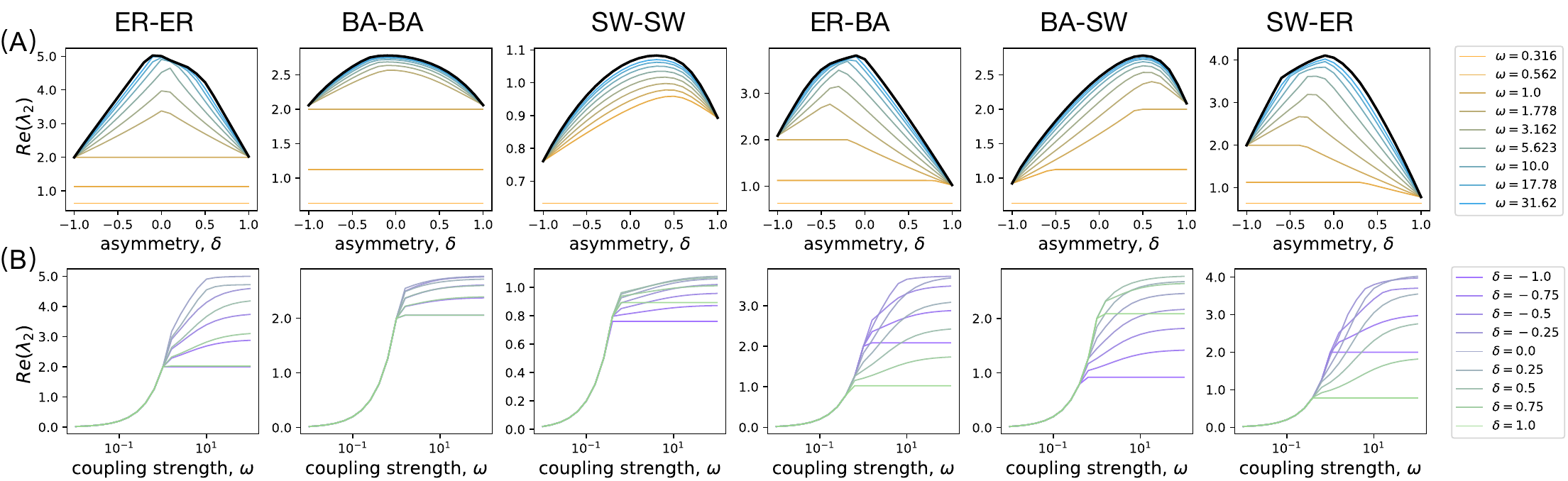}
	\caption{ 
	{\bf  Cooperative optima for random multiplex networks with layers sampled from   random-graph models.}
	We study $\text{Re}(\lambda_2)$ versus either (A) $\delta$ or (B) $\omega$ for six multiplex networks, each having $T=2$ layers  coupled by the interlayer Laplacian matrix defined in \eqref{eq:LI_2}. Each network layer is sampled from one of three generative models:	
	Erd\"os-R\'enyi graphs (ER), Barab\'asi-Albert graphs (BA), and Watts-Strogatz small-world graphs (SW). 
Black curves in (A) depict a theoretical prediction for large $\omega$ that we will develop in Sec.~\ref{sec:theory}. 
	}
	\label{fig:behaviors_RG}
\end{figure*}

We conclude this section by studying $\text{Re}(\lambda_2)$ for random multiplex networks with $T=2$ layers, each of which is generated by two of the following generative models:

\begin{itemize}
\item[(i)]   Erd\"os-R\'enyi (ER) graphs \cite{erdds1959random} in which each edge   
is independently created as a Bernoulli random variable 
probability $p$. 
\item[(ii)]   Barab\'asi-Albert (BA) graphs \cite{barabasi1999emergence}, which are grown by attaching each new node to $m$ edges that are preferentially selected based on their  degree. 
\item[(iii)]   Watts-Strogatz small-world (SW) graphs \cite{watts1998collective}, which are created by first assigning nodes positions along a ring and by creating edges between each node and its $k$ nearest neighbors. Then, each edge is replaced by a new, randomly selected edge with probability $p$.
\end{itemize}
We selected these three models, since they give rise to different well-known properties: degree-homogeneity for ER graphs; degree-heterogeneity for BA graphs; and the small-world property for SW graphs.
Unless otherwise specified, we construct multiplex networks with $N=500$ nodes and choose $p=0.02$ for the ER   layers,  $m=4$ edges for the BA   layers, and $(p,k)=(0.2,8)$ for  the  SW   layers. Additional parameter choices are studied in Appendix~\ref{sec:appendix_large}.




In Fig.~\ref{fig:behaviors_RG},  we study $\text{Re}(\lambda_2)$ for six multiplex networks in which each layer is generated by one of the three generative models.  (For example, ``ER-BA'' indicates that the first layer is an ER graph, whereas the second is a BA graph.) 
Similar to Figs.~\ref{fig:behaviors_toy}(B) and (C), we   plot   $\text{Re}(\lambda_2)$ versus $\delta$ in Fig.~\ref{fig:behaviors_RG}(A) and   $\text{Re}(\lambda_2)$ versus    $\omega$ in Fig.~\ref{fig:behaviors_RG}(B).
Our main observation in  Fig.~\ref{fig:behaviors_RG}(A) is that these six random multiplex networks all exhibit a cooperative optimum for sufficiently large $\omega$.


In Sec.~\ref{sec:theory}, we will develop and apply theory to shed light on  the various dynamical and structural mechanisms that can give rise to behaviors (i)--(v) and yield optimally coupled systems that are either cooperative or non-cooperative. Our analytical approach is motivated by the following observations. Observe  in each panel of Figs.~\ref{fig:behaviors_toy}(B) and  \ref{fig:behaviors_RG}(A) that the  different curves   reflect different choices for the interlayer coupling strength $\omega$ and that the  qualitative effects on $\text{Re}(\lambda_2)$ of $\delta$  are consistent across a wide range of  $\omega$ values. Moreover, observe in each panel of Figs.~\ref{fig:behaviors_toy}(C) and  \ref{fig:behaviors_RG}(B) that there exists a critical value of  $\omega$ below which $\delta$ has no observable effect on $\text{Re}(\lambda_2)$. That is, the various effects of coupling asymmetry only arise when $\omega$ is sufficiently large. [Interestingly, the fifth system in  Fig.~\ref{fig:behaviors_toy} is an exception, since $\delta$ appears to affect   $\text{Re}(\lambda_2)$ only for intermediate values of $\omega$ near the  ICO peak.] 
With this in mind, in the next section we present theory to predict the effects of coupling asymmetry on $\text{Re}(\lambda_2)$ for when the layers are strongly coupled  with large $\omega$.

\section{Theoretical Results}\label{sec:theory}

We now present our main theoretical findings.
In Sec.~\ref{sec:thy}, we analyze $\text{Re}(\lambda_2)$ in the limit of strong interlayer coupling.
%
In Sec.~\ref{sec:exist}, we build on these results to provide criterion that guarantees the existence of a cooperative optimum. Our derivations are defered to appendices.

\begin{figure*}[h]
    \centering 
    \includegraphics[width=.99\linewidth]{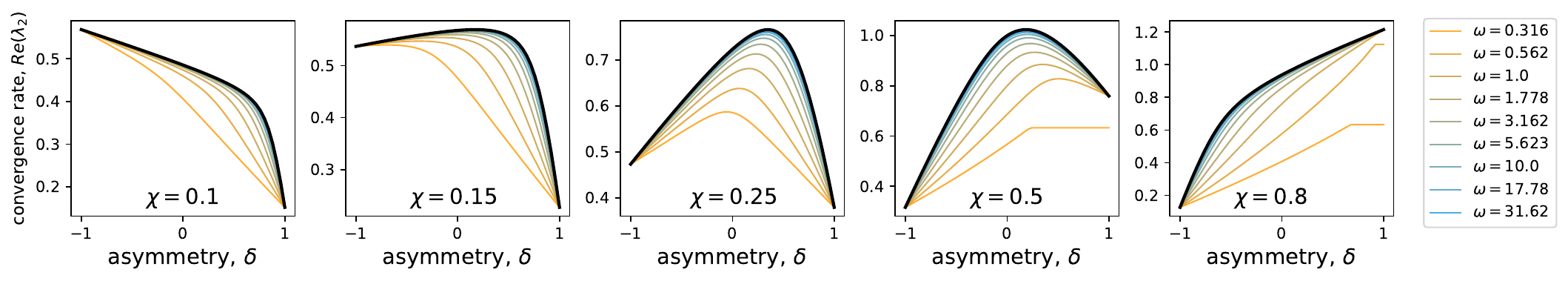}    
    \caption{{\bf Layers' relative timescales influence if optima are cooperative vs. non-cooperative.}
    We plot the convergence rate $\text{Re}(\lambda_2)$ versus $\delta$ for the system that was  visualized in the second column of Fig.~\ref{fig:behaviors_toy}(A).  Different columns correspond to different choices for a rate-scaling parameter $\chi\in(0,1)$, which we introduce to tune whether layer 1 is much faster ($\chi\approx 1$) or layer 2 is much faster ($\chi\approx 0$). 
    Colored curves yield results for different $\omega$ (see legend), and black curves depict our theoretical prediction for large $\omega$ that is given by \eqref{eq:lim_lam2}--\eqref{eq:L_bar_2}.
By comparing across the columns, observe that their optima are cooperative for intermediate values of $\chi$ (i.e., the second, third and fourth columns) and  non-cooperative when $\chi$ is sufficiently small or large (see left-most and right-most columns).
    }
    \label{fig:timescales_Simple}
\end{figure*}

\subsection{Theory for Strong Coupling,  $\omega\to\infty$.}\label{sec:thy}

To provide  analytical guidance, we  characterize the dependence of $\text{Re}(\lambda_2)$ on $\delta$ using spectral perturbation theory  for  multilayer networks with asymmetric matrices \cite{taylor2019tunable,taylor2020multiplex}. We present the derivations in Appendix \ref{sec:appendix_A}  and summarize our findings here. The black curves in Fig.~\ref{fig:behaviors_toy}(B)  and Fig.~\ref{fig:behaviors_RG}(A) depict our predictions for large $\omega$,
\begin{align}\label{eq:lim_lam2}
    \lim_{\omega\to\infty} \lambda_2  &= \overline{\lambda}_2(\delta),
\end{align}
where $\overline{\lambda}_2(\delta)$ is the   eigenvalue  of $\overline{\mathbf{L}}(\delta)$ that has the second-smallest real part, and
\begin{eqnarray}\label{eq:L_bar}
    \overline{\mathbf{L}}(\delta)
    &=&   \sum_{t=1}^{T} w_{t}(\delta)    \mathbf{L}^{(t)}
\end{eqnarray}
is a weighted average of the layers' Laplacian matrices. The weights $w_t (\delta) = u_t(\delta)/ {\sum_{t'} u_{t'}(\delta)}$ come from the entries of the left eigenvector $\bm{u}(\delta)=[u_1(\delta),\dots,u_T(\delta)]^\mathrm{T}$ that is associated with the zero-valued (i.e., trivial) eigenvalue of    $\bm{L}^\textrm{I}(\delta)$. (See \cite{gomez2013diffusion,sole2013spectral} for results that are similar to \eqref{eq:L_bar} but which assume symmetric coupling, $\delta=0$).

Equations~\eqref{eq:lim_lam2}--\eqref{eq:L_bar} imply that when a multiplex consensus system is strongly coupled, the convergence rate is identical to that for consensus on an ``effective'' network  that is associated with a Laplacian matrix $\overline{\mathbf{L}}(\delta)$, and the effects of $\delta$ can be examined by considering the   dependence of $w_t(\delta)$ on $\delta$. For example, for $T=2$ layers, the interlayer Laplacian $\bm{L}^\textrm{I}(\delta)$ is given by \eqref{eq:LI_2},  $\bm{u}(\delta)=[ 1+\delta ,1-\delta ]^\mathrm{T}$ and we find
\begin{align} \label{eq:L_bar_2}
    \overline{\mathbf{L}}(\delta) = \left( \frac{1+\delta}{2} \right) \mathbf{L}^{(1)} + \left(\frac{1-\delta }{2}\right) \mathbf{L}^{(2)}.
\end{align}

Despite this simple form,  the associated convergence rate $\text{Re}(\overline{\lambda}_2(\delta) )$ can exhibit a   complicated dependence on $\delta$. For example, observe in Fig.~\ref{fig:behaviors_toy}(B) and Fig.~\ref{fig:behaviors_RG}(A) that in addition to being  accurate for large $\omega$, this  theory  predicts the qualitative behavior of the relationship between $\text{Re}(\lambda_2)$ and $\delta$  for a  broad range of $\omega$. 
At the same time, also observe that the characterization of the optimum as being cooperative vs. non-cooperative in the limit $\omega\to\infty$ also is predictive of that optimum for other choices of $\omega$.
%


One consequence of \eqref{eq:lim_lam2}--\eqref{eq:L_bar_2} is the following limiting behavior for when either layer 1 or layer 2 directs the other layer without feedback:
\begin{subequations}\label{eq:limlim}
\begin{align}
    \lim_{\delta\to1} \overline{\lambda}_2(\delta)  &\to {\lambda}_2^{(1)}, \\
    \lim_{\delta\to - 1} \overline{\lambda}_2(\delta)  &\to {\lambda}_2^{(2)},
\end{align}
\end{subequations}
where $\lambda_2^{(1)}$  and $\lambda_2^{(2)}$ are the second-smallest eigenvalues of $\mathbf{L}^{(1)}$ and $\mathbf{L}^{(2)}$, respectively. This follows immediately after considering that $\lim_{\delta\to 1}\overline{\mathbf{L}}(\delta) = \mathbf{L}^{(1)}$ and $\lim_{\delta\to -1}\overline{\mathbf{L}}(\delta) = \mathbf{L}^{(2)}$. Thus, the convergence rate of a single   layer  dictates the overall system's convergence when the coupling asymmetry implements non-cooperation.

Given the above theory, we can now better understand the results that were previously shown in Fig.~\ref{fig:behaviors_toy}(B). Recall that the black curves in that figure represent our analytical prediction   $\text{Re}(\overline{\lambda}_2(\delta))$. Moreover, for all five systems, we predict that $\text{Re}(\lambda_2 )$ converges to $\text{Re}(\lambda_2^{(2)})$ and $\text{Re}(\lambda_2^{(1)})$ in the limits $\delta \to-1$ and $\delta \to1$, respectively.
Interesting, our derivation of \eqref{eq:limlim} has assumed the limit of large $\omega$, but one can observe in Fig.~\ref{fig:behaviors_toy}(B) that it accurately predicts the limiting $\delta\to\pm1$ behavior of $\text{Re}(\lambda_2)$ for a wide range of coupling strengths $\omega$. That is, most of the  curves that represent different $\omega$ values converge to the same point on the left-hand and right-hand sides of each subpanel. We only observe \eqref{eq:limlim} to yield an inaccurate prediction for the  $\delta \to \pm 1$ limits of $\lambda_2$ when $\omega$ is very small [e.g., when $\omega=0.316$ in columns 1, 3 and 4 of  Fig.~\ref{fig:behaviors_toy}(B) or $\omega=\{0.316, 0.562\}$ in column 2 of Fig.~\ref{fig:behaviors_toy}(B)].
Similar results can also be observed in Fig.~\ref{fig:behaviors_RG}(A). For all six   multiplex networks, $\text{Re}(\lambda_2 )$ converges to $\text{Re}(\lambda_2^{(2)})$ and $\text{Re}(\lambda_2^{(1)})$ in the limits $\delta \to-1$ and $\delta \to1$, respectively, for a wide range of coupling strengths $\omega$.

Equation~\eqref{eq:limlim} also hints at why an optimally coupled system will exhibit layer dominance rather than a cooperative optimum. Consider the first column of Fig.~\ref{fig:behaviors_toy}, where we observed a non-cooperative optimum: $\text{Re}(\lambda_2)$ obtains its maximum at $\delta=-1$, whereby layer 2 influences layer 1 without feedback. In this case, layer 1 is an undirected chain graph that has a convergence rate of $\lambda_2^{(1)}\approx0.26$, whereas layer 2 is a star graph that has a convergence rate of $\lambda_2^{(2)}\approx1$. 
Thus, the convergence rate is faster for the star graph than the chain graph, and this particular system converges fastest when  the faster system (layer 2) non-cooperatively influences the slower one without feedback. 
%
In the next section, we develop theory that can help determine whether the optimum is cooperative and how that relates to the layers' individual timescales.

\subsection{ Existence Guarantee for a  Cooperative Optimum}\label{sec:exist}

We first show that a cooperative  optimum  vs. layer dominance occurs when the layers have a sufficiently similar timescale, and in fact, we can adjust a system between these two behaviors by varying their   timescales. Recall from the model definition in Sec.~\ref{sec:model2}
that we can use a \emph{rate-scaling parameter} $\chi\in(0,1)$ to vary the relative convergence rate for each layer.
That is, we define the the intralayer Laplacians 
    $\chi \mathbf{L}^{(1)} $
    and $ (1-\chi) \mathbf{L}^{(2)}$ so that their separate convergence rates are $\chi \text{Re}(\lambda_2^{(1)})$ and $(1-\chi) \text{Re}(\lambda_2^{(2)})$, respectively. It also then follows that \eqref{eq:L_bar_2} takes the form $\overline{\mathbf{L}}(\delta) = \frac{1 + \delta}{2}  \chi\mathbf{L}^{(1)} + \frac{1 - \delta}{2}  (1-\chi)\mathbf{L}^{(2)}$.

In Fig.~\ref{fig:timescales_Simple}
, we plot the convergence rate $\text{Re}(\lambda_2)$ versus   $\delta$ for 
 the system that was  visualized in the second system in Fig.~\ref{fig:behaviors_toy}(A).
 Different columns correspond to different choices for $\chi$. In each panel, different curves reflect different choices for $\omega$. Black curves indicate the analytical prediction for large $\omega$ given by \eqref{eq:L_bar}. By comparing across the columns of Fig.~\ref{fig:timescales_Simple}, observe  that their optima are cooperative for intermediate values of $\chi$ (i.e., the second, third and fourth columns) and  non-cooperative when $\chi$ is sufficiently small or large (see the left-most and right-most columns).  That is, we observe a cooperative optimum for these systems when the convergence rates for the two separate layers are sufficiently similar. Otherwise, in this experiment, we find that the convergence is fastest when the faster layer influences the slower one without a reciprocated influence, i.e., layer dominance. 
This observation is further supported  in Appendix \ref{sec:appendix_large}, where we provide additional figures similar to Fig.~\ref{fig:timescales_Simple} for the family of random multiplex networks described in Sec.~\ref{sec:large}.


We now present a criterion that predicts how  the optimal asymmetric coupling of layers can change between  cooperative and non-cooperative as one varies   $\chi$. We present this theory in  Appendix \ref{sec:appendix_B} and summarize our main findings here. We predict the existence/nonexistence of an optimum in the limit of large $\omega$ by considering the derivative
\vspace{-.01in} 
\begin{align}\label{eq:lag}
    \overline{\lambda}_2'(\delta) \equiv \frac{d}{d\delta} \text{Re} \left( \overline{\lambda}_2(\delta) \right) 
\end{align} 
and by invoking Rolle's Theorem \cite{sahoo1998mean} for a continuous function: if $\overline{\lambda}_2'(-1)>0$ and $\overline{\lambda}_2'(1)<0$,  then there exists at least one   optimum $\hat{\delta}=\text{argmax}_\delta \text{Re}(\lambda_2)$  that is cooperative, i.e., $\hat{\delta} \in(-1,1)$. In the limits $\delta\to \pm1$,  the derivatives   $\overline{\lambda}_2'(\delta) $  converge  to a simplified form:
\begin{subequations}\label{eq:lam2_prime}
\begin{align}
    \overline{\lambda}_2^{\prime}(1) &=
     \frac{-{\mathbf{u}^{(1)}}^* \mathbf{L}^{(2)} \mathbf{v}^{(1)}} {2{\mathbf{u}^{(1)}}^* \mathbf{v}^{(1)}} 
     + \chi 
     \frac{{\mathbf{u}^{(1)}}^*
     \left(\mathbf{L}^{(1)} + \mathbf{L}^{(2)}\right) \mathbf{v}^{(1)}} 
         {2{\mathbf{u}^{(1)}}^* \mathbf{v}^{(1)}},            \\
     \overline{\lambda}_2^{\prime}(-1)& =  \frac{-\lambda^{(2)}_2}{2}
     + \chi 
     \frac{{\mathbf{u}^{(2)}}^* 
     \left(\mathbf{L}^{(1)} +\mathbf{L}^{(2)} \right) \mathbf{v}^{(2)}} 
     {2{\mathbf{u}^{(2)}}^*  \mathbf{v}^{(2)}}.
\end{align}
\end{subequations}
Letting $t\in\{1,2\}$, here we define  $\mathbf{u}^{(t)}$ and $\mathbf{v}^{(t)}$ as the left and right eigenvectors, respectively, that are associated with the  eigenvalue $\lambda^{(t)}_2$ of $\mathbf{L}^{(t)}$ that has the second-smallest real part (assumed to be nonzero). Symbol $*$ denotes a vector's complex conjugate. Interestingly, \eqref{eq:lam2_prime} imply that the derivatives $\overline{\lambda}_2^{\prime}(\delta)$ at $\delta =\pm 1$ change linearly with the time-scaling parameter $\chi$.

By combining \eqref{eq:lam2_prime} with Rolle's Theorem, we can identify for each $\chi$, whether a cooperative optimum is guaranteed to exist. Moreover, we can use the linear form of  \eqref{eq:lam2_prime} to predict the values of $\chi$ at which $ \overline{\lambda}_2^{\prime}(1) $ and  $\overline{\lambda}_2^{\prime}(-1) $ change sign, allowing us to obtain a simplified criterion for this trait. That is, we solve for $\text{Re} \left( \overline{\lambda}_2^{\prime}(\pm 1) \right)=0$ to obtain x-intercepts for $\chi$ 
\begin{subequations}\label{eq:critical_chi}
\begin{align}
    \hat{\chi}(1) &=   \text{Re} \left( \frac{ {\mathbf{u}^{(1)}}^* \mathbf{L}^{(2)} \mathbf{v}^{(1)}} { {\mathbf{u}^{(1)}}^* (\mathbf{L}^{(1)} + \mathbf{L}^{(2)}) \mathbf{v}^{(1)}} \right), \\
    \hat{\chi}(-1) &= \text{Re} \left( \frac{ {\mathbf{u}^{(2)}}^* \mathbf{L}^{(2)} \mathbf{v}^{(2)}} 
        { {\mathbf{u}^{(2)}}^* (\mathbf{L}^{(1)} + \mathbf{L}^{(2)}) \mathbf{v}^{(2)}} \right) .
\end{align}
\end{subequations}
Considering the lines defined in \eqref{eq:lam2_prime}, when the slopes are positive and y-intercepts are negative (which we observe to be true for all our experiments, but in principle, it may not always be true), then the implications of Rolle's Theorem can be summarized in a simplified form: a cooperative optimum is guaranteed to exist when the the two layers have  sufficiently similar timescales in that
\begin{align}\label{eq:guar}
    \chi \in (\hat{\chi}(-1),\hat{\chi}(1)).
\end{align}
That is, $\chi$ is neither too large nor too small. Moreover, the criterion given by \eqref{eq:guar}  also guarantees that the convergence rate $\text{Re}(\overline{\lambda}_2(\delta)) $ of the multiplexed system is faster than that for either system,
\begin{align}\label{eq:gaa}
    \text{Re}(\overline{\lambda}_2(\delta)) >  \max\{\chi \text{Re}({\lambda}^{(1)}_2), (1-\chi)\text{Re}({\lambda}^{(1)}_2) \}.
\end{align}

\begin{figure}[!t]
	\centering
    	\includegraphics[width=\linewidth]{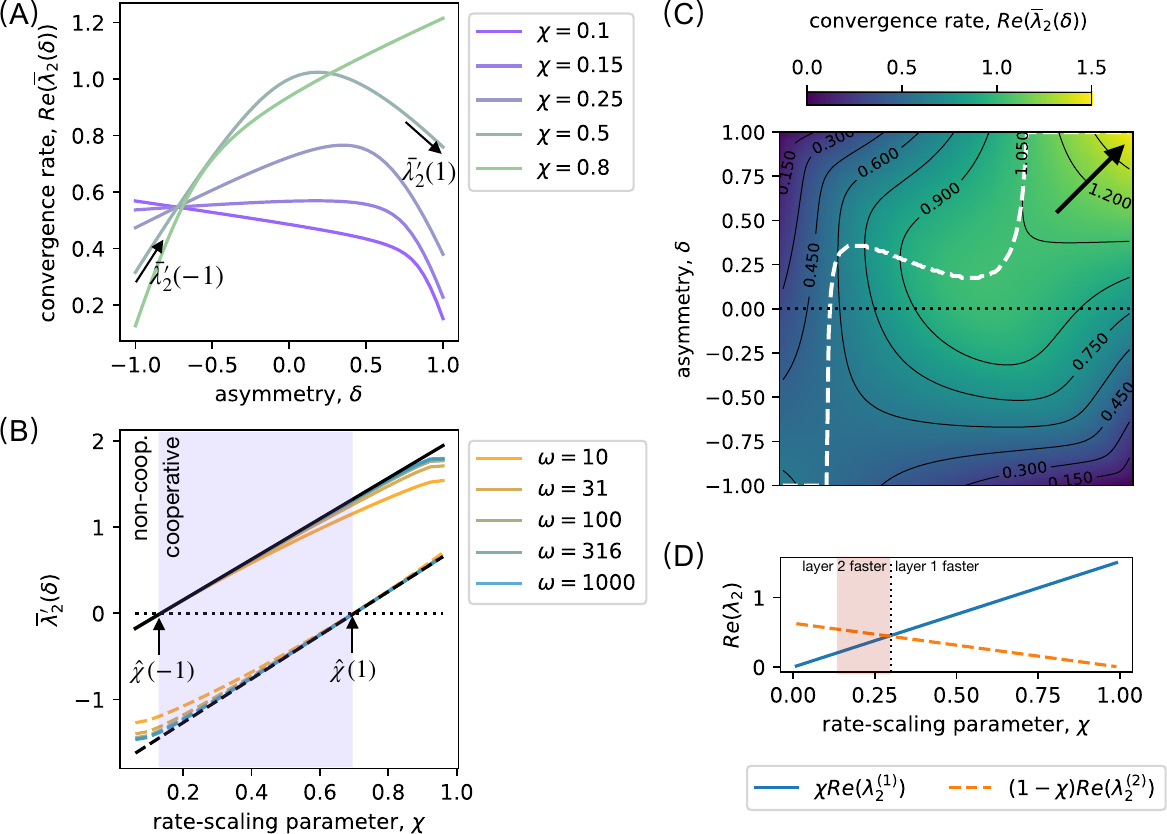}
	\caption{ 
	{\bf Combined effects of coupling asymmetry and relative timescale for the second system in Fig.~\ref{fig:behaviors_toy}(A).}
	(A) We plot the predicted convergence rate $\text{Re} \left(\overline{\lambda}_2(\delta)\right)$ given by \eqref{eq:lim_lam2}--\eqref{eq:L_bar_2}, which assume    large  coupling strength $\omega$, versus  $\delta$ for several choices of   $\chi$.   Black arrows highlight derivative $\overline{\lambda}_2'(\delta) \equiv\frac{d}{d\delta} \text{Re} \left(\overline{\lambda}_2(\delta) \right)$ for one curve at $\delta = \pm 1$.
	(B)~Black solid and dashed lines indicate our predicted derivatives   $\overline{\lambda}_2'(\pm1)$ given by \eqref{eq:lam2_prime}, and colored curves indicate observed values of $ \frac{d}{d\delta} \text{Re}({\lambda}_2)$ for several   choices of $\omega$. As indicated by the shaded region, Rolle's Theorem guarantees the existence of a cooperative optimum   when    $\chi\in(\hat{\chi}(-1),\hat{\chi}(1))$ [see  \eqref{eq:guar}], where $\hat{\chi}(\pm1)$ are given by  \eqref{eq:critical_chi} and are  the values of $\chi$ at which   $\overline{\lambda}_2'(\pm1)$ change sign.
	(C)~We use color to depict  $ \text{Re} \left(\overline{\lambda}_2(\delta) \right)$ across the  $(\delta,\chi)$ parameter space.  The dashed white curve shows the asymmetry $\hat{\delta}$ that maximizes $\text{Re} \left(\overline{\lambda}_2(\delta) \right)$ for each value of $\chi$. The arrow in the top-right corner highlights  the location of the overall optimum,  $\max_{\delta,\chi} \text{Re}\left(\overline{\lambda}_2(\delta)\right)$, which is a non-cooperative optimum at $ ( {\delta}, {\chi})  \approx(1,1)$.
    	(D)~We plot the  convergence rate for each separate layer as a function of $\chi$ and highlight a counter-intuitive phenomenon that occurs for the shaded range of $\chi$ (see text).
	}
	\label{fig:optima_Simple}
\end{figure}


\begin{figure*}[!t]
	\centering
	\includegraphics[width=1\linewidth]{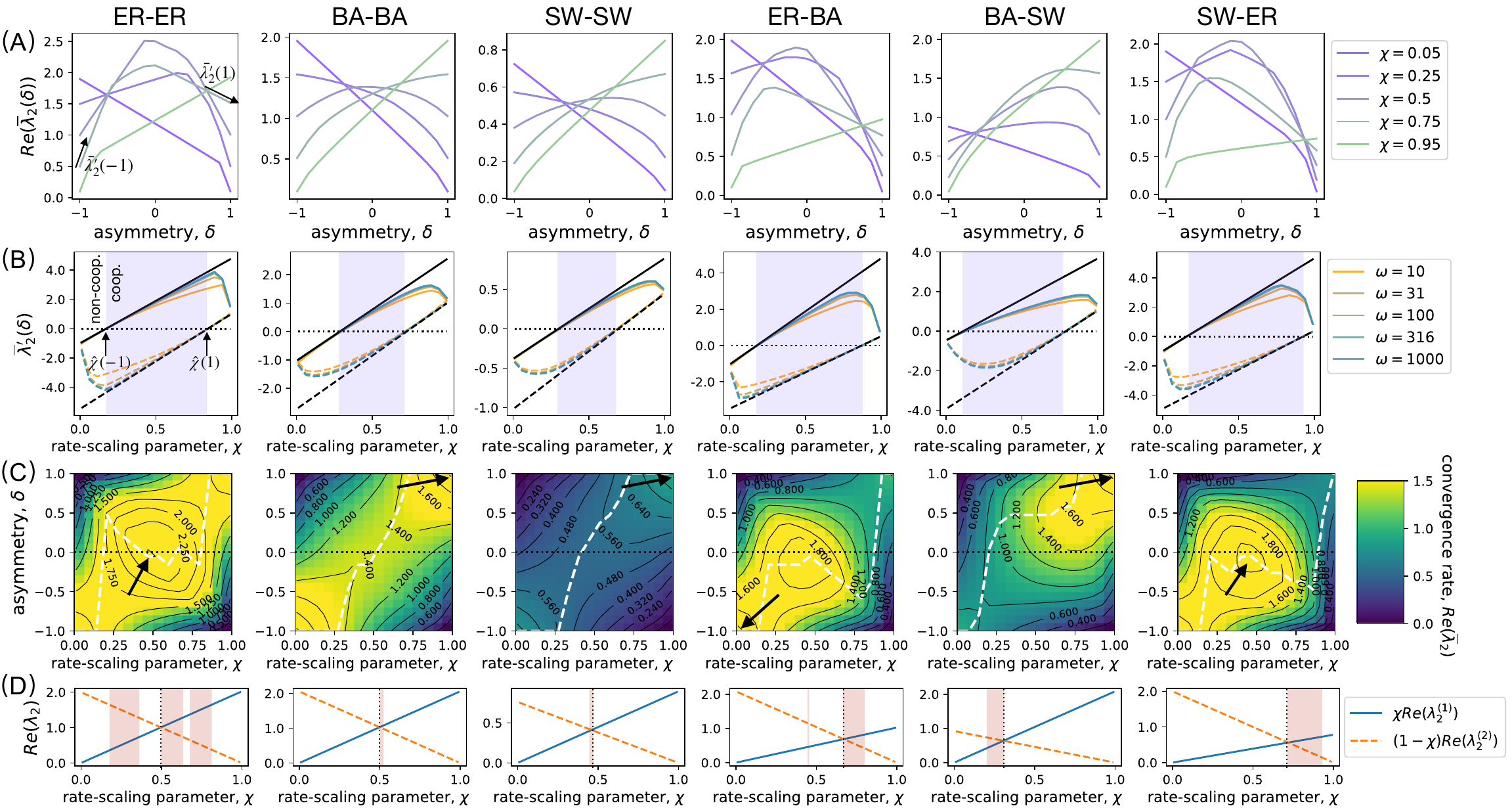}
	\caption{ 
	{\bf Combined effects of coupling asymmetry $\delta$ and relative timescale $\chi$ on the convergence rates for random multiplex networks.}
	Here, we extend the results shown in Fig.~\ref{fig:optima_Simple} to the six random multiplex networks   described in Sec.~\ref{sec:large}.
	The main difference is highlighted by the black arrows in panel (C), which denote the locations of the overall optimum: $\max_{\delta,\chi} \text{Re}\left(\overline{\lambda}_2(\delta)\right)$. 
	Observe that the overall optimum is   cooperative   for the ER-ER and SW-ER networks (i.e., see the leftmost and rightmost columns), 	whereas it is non-cooperative for the others. Additionally,   as indicated by the shaded regions   in (D), there are     several   counter-intuitive scenarios.  For these value of $\chi$, convergence occurs fastest when the  layer with slower dynamics more-strongly influences the faster layer.
	}
	\label{fig:optima_RG}
\end{figure*}

In Fig.~\ref{fig:optima_Simple}, we study the combined effects of the  asymmetry parameter $\delta$ and rate-scaling parameter $\chi$   on the convergence rate for the multiplex network that was visualized in the second column of Fig.~\ref{fig:behaviors_toy}(A). In Fig.~\ref{fig:optima_Simple}(A), we plot  $\text{Re} \left(\overline{\lambda}_2(\delta) \right)$ versus $\delta$ for several choices of $\chi$. Observe  for $\chi\in \{0.15,0.25,0.5\}$ that the optimum is cooperative, which agrees with our predicted range   given by \eqref{eq:guar}. This range is indicated by the shaded region in Fig.~\ref{fig:optima_Simple}(B), where  solid and dashed black lines  illustrate \eqref{eq:lam2_prime}, and their intersections with the x-axis yield the critical values of $\chi$ that are defined by \eqref{eq:critical_chi}. The colored lines  depict empirically observed values of $\frac{d}{d\delta} \text{Re}(\lambda_2)$ that were computed directly from $\mathbb{L}(\omega,\delta)$ with several choices of coupling strength  $\omega$. Interestingly, our theoretical predictions for $\hat{\chi}(\pm 1)$   appear to describe the x-axis intercepts for all of these $\omega$ values, which is surprising since our derivation has assumed  the limit large  $\omega$.

In Fig.~\ref{fig:optima_Simple}(C), we visualize by color the convergence rate $\text{Re}\left(\overline{\lambda}_2(\delta)\right)$ across the parameter space $\delta\in[-1,1]$ and $\chi\in[0,1]$. The dashed white  curve   indicates the optimal asymmetry $\hat{\delta} = \text{argmax}_\delta \text{Re}\left(\overline{\lambda}_2(\delta)\right)$ for each value of $\chi$. The   arrow in the top-right corner highlights   the location of the overall optimum, $ \text{max}_{\delta,\chi} \text{Re}\left(\overline{\lambda}_2(\delta)\right)$, which is a non-cooperative optimum at $(\delta,\chi)=(1,1)$. That is, convergence is optimally fast for this system when the first layer is made to be as fast as possible (${\chi}=1$), and it influences the second layer without a reciprocated influence (${\delta}=1$).

In Fig.~\ref{fig:optima_Simple}(D), we highlight a counter-intuitive property: when  two layers are coupled with an optimal asymmetry parameter $\hat{\delta}$, it may involve the slower system having a greater influence over the  faster system. To make this point, the solid blue  and dashed orange lines in Fig.~\ref{fig:optima_Simple}(D) indicate the separate convergence rate of each   (uncoupled) layer. The vertical dotted line in Fig.~\ref{fig:optima_Simple}(D) highlights that layer 1 is slower than layer 2 when $\chi <0.29$. Moreover, observe in Fig.~\ref{fig:optima_Simple}(C) that $\hat{\delta}>0$ (see dashed white  curve) for approximately $\chi>0.13$. Therefore, the shaded region  $\chi\in(0.13,0.29)$  in Fig.~\ref{fig:optima_Simple}(D) highlights values of $\chi$ in which the optimal asymmetry correspond to when  the slower system (layer 1)   more strongly influences the faster one (layer 2). That said,  our experiments also suggest that if $\chi$ is sufficiently large or small (i.e., of one layer is much, much faster than the other), then the optimal asymmetry does coincide with our intuition:  the faster system should more strongly influence the slower one.

In Fig.~\ref{fig:optima_RG}, we repeat the  experiment  shown in Fig.~\ref{fig:optima_Simple}, except we now study the six random multiplex networks that were introduced in Sec.~\ref{sec:large}.
In Fig.~\ref{fig:optima_RG}(A), we plot $\overline{\lambda}_2(\delta)$ versus $\delta$ for several choices of $\chi$, and observe for all networks   that the optimum is cooperative when $\chi$ is neither too large nor too small. Specifically, this occurs when  $\chi\in \{0.25,0.5,0.75\}$ for the multiplex networks labeled ER-ER, ER-BA, BA-SW and SW-ER, and when $\chi = 0.5 $ for the BA-BA and SW-SW networks. The remaining curves depict  a non-cooperative optimum (i.e., layer dominance).
These values of $\chi$ are in excellent agreement  with our predicted ranges of $\chi$ that yield a cooperative optimum, which are given by \eqref{eq:guar} and are   indicated by the shaded regions in Fig.~\ref{fig:optima_RG}(B). As before, the  solid and dashed black lines   illustrate \eqref{eq:lam2_prime}. 

In Fig.~\ref{fig:optima_RG}(C), we use color to visualize     $\text{Re}\left(\overline{\lambda}_2(\delta)\right)$ across the parameter space $\delta\in[-1,1]$ and $\chi\in[0,1]$. In each column, dashed white curves highlight the optimum $\hat{\delta} = \text{argmax}_\delta \text{Re}\left(\overline{\lambda}_2(\delta)\right)$ for a fixed, given value of $\chi$, whereas the arrows highlight the   overall optimum, $ \text{max}_{\delta,\chi} \text{Re}\left(\overline{\lambda}_2(\delta)\right)$. 
Observe that the overall optimums are cooperative for the models ER-ER and SW-ER, which occur, respectively, near $ ( {\delta}, {\chi})  \approx(-0.1, 0.48)$ and $( {\delta}, {\chi})  \approx(-0.03, 0.48)$.
Intuitively,  consensus is maximally accelerated in both models when the second layer  is slightly more influential than the first (i.e., since $ {\delta}<0$), and the first layers' dynamics are slightly slowed down (i.e., since $ {\chi}<0.5$).
On the other hand, the overall optimum is non-cooperative for the models BA-BA, SW-SW and BA-SW---that is, convergence is optimally fast for these systems when the first layer is made to be as fast as possible ($\chi\to1$) and that layer influences the second without reciprocation ($\delta\to1$). 
The model   ER-BA also exhibits a non-cooperative overall optimum; however in this case the optimum involves the second layer dominating the first: $ ( {\delta}, {\chi}) \to (-1,0) $.
%

Finally, lines in Fig.~\ref{fig:optima_RG}(D) illustrate each  layers' separate convergence rate, and the shaded regions highlight values of $\chi$  in which the dynamics are  counter intuitive. For these values of $\chi$, convergence is fastest when the layer with slower dynamics more-strongly influences the faster layer.

\section{Application: Modeling the  Collective Decisions  of   Human-AI  Systems}\label{sec:AI}

We propose the interconnected consensus system presented in Sec.~\ref{sec:model} as an insightful model for studying the collective decisions of   human-AI   teams. In Sec.~\ref{sec:motiv}, we motivate and interpret our model for this application.
In Sec.~\ref{sec:time_balance}, we study the network presented in Fig.~\ref{fig1}(A), showing that  a cooperative optimum requires     human-human  and  AI-AI interactions to have  similar timescales.
In Sec.~\ref{sec:overall}, we study the overall optimum that maximizes the convergences rate by simultaneously tuning layers' timescales and  the asymmetry of coupling between the human layer  and the AI layer.

\begin{figure*}[!ht]
    \centering 
    \includegraphics[width=.99\linewidth]{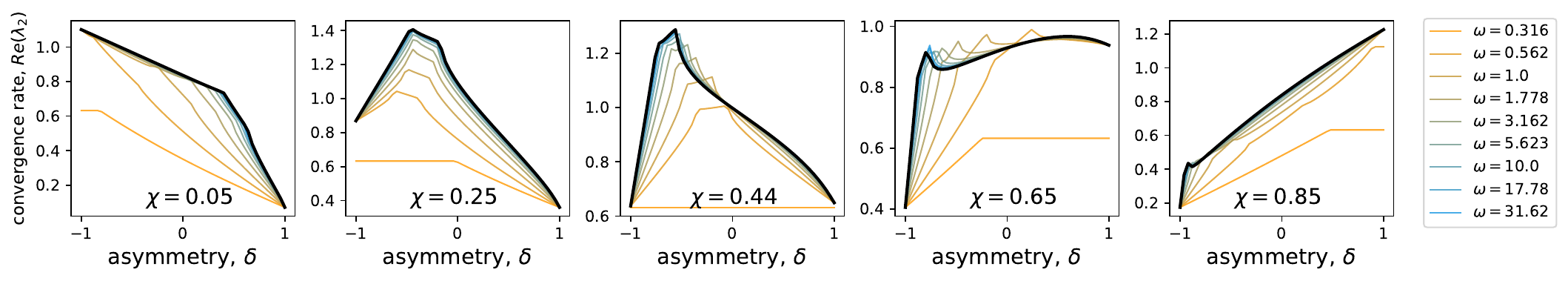}
    \caption{
     {\bf A cooperative optimum requires the human and AI network layers   have   sufficiently similar timescales.}
    For the system   visualized in Fig.~\ref{fig1}(A), we plot the convergence rate $\text{Re}(\lambda_2)$ versus $\delta$.  Different columns reflect different choices for the  rate-scaling parameter $\chi\in(0,1)$, which   tunes whether layer 1 is much faster ($\chi\approx 1$) or layer 2 is much faster ($\chi\approx 0$). 
    Colored curves yield results for different $\omega$ (see legend), and black curves depict our theoretical prediction for large $\omega$ that is given by \eqref{eq:lim_lam2}--\eqref{eq:L_bar_2}.
By comparing across the columns, observe that their optima are cooperative for intermediate values of $\chi$ (i.e., the second, third and fourth columns) and  non-cooperative when $\chi$ is either too small or large (see left-most and right-most columns).
    }
    \label{fig:timescales_HumanAI}
\end{figure*}

\subsection{Motivation and Interpretation}\label{sec:motiv}

Within the social and cognitive sciences, there is a rich literature of dynamical models for collective decision making by social groups and social networks, many of which describe processes by which populations can reach consensus \cite{hinsz1990cognitive,fiol1994consensus,flood2000chief,mohammed2001toward}.
%
Of particular interest is the Abelson model \cite{abelson1964mathematical,abelson1967mathematical} for opinion dynamics, which takes the form of a linear differential equation similar to our proposed model \eqref{eq:con}.
Similar consensus models have also been utilized in the biology community to model decision making by animal groups \cite{conradt2005consensus,westley2018collective} and by the computer science community  to implement decentralized algorithms for machine learning and AI \cite{bijral2017data,assran2019stochastic,niwa2020edge,vogels2020powergossip,kong2021consensus,huynh2021}.
We offer one interpretation of the latter application---that is, a set of ML/AI models are trained on different data to have different parameter values, and they reach a ``collective decision'' on the best model parameters using a similar model for consensus.

Given the ubiquity of consensus models for  collective decision making in a wide variety of applications, we propose the interconnected consensus system   in Sec.~\ref{sec:model}  as a simple-yet-informative  model for   collective decisions made by a social network in which individual are supported by AI agents, who themselves coordinate and collectively learn.
In this context,  the asymmetry parameter has the following interpretation:  $\delta>1$ implies that the humans' states more strongly influence those of the AI agents, and $\delta<1$ implies the opposite. Moreover, the rate-scaling parameter $\chi$ controls the relative timescale for coordination via human-human interactions as compared to AI-AI interactions.  Collective consensus-based decisions  within the social network is  represented by the intralayer consensus model $\frac{d}{d\tau} {\bf x}^{(1)}(\tau) = - \chi \mathbf{L}^{(1)}  {\bf x}^{(1)} (\tau)$, while consensus among AI agents  is represented by $\frac{d}{d\tau} {\bf x}^{(2)}(\tau) = - (1-\chi )\mathbf{L}^{(2)}  {\bf x}^{(2)} (\tau)$. The choice $\chi\approx 0$ corresponds to when the AI agents coordinate much faster than the humans, while $\chi\approx 1$ implies the opposite. 

The two consensus systems are coupled according to \eqref{eq:con} so that consensus over the entire system represents a  decision that is collectively obtained over the `multiplexed' human-AI social system. Such a model could interpreted as a binary model so that  a state $x_p\in \mathbb{R}$ represents the tendency of a human or agent to make some particular binary decision, such as   taking a strategic military action \cite{rasch2003incorporating,azar2021drone,zhou2019bayesian} or investing in a particular stock  \cite{albadvi2007decision,chou1997stock,yu2005designing}. 
Having a strong preference for (or against) such an action would be represented by a large positive (or negative) value, and weaker preferences can be represented by small-magnitude values. We interpret a collective decision of yea or nay  as the converged state being positive or negative, and it can be beneficial for systems to make optimally fast decisions (which can be engineered by maximizing the convergence rate). That said, in real-world scenarios one should also consider other system properties that are essential including, e.g., trust \cite{ezer2019trust} and coordinating agents'   expertise to be complementary \cite{bansal2019beyond}.

We now further study the multiplex network shown in Fig.~\ref{fig1}(A), where layer 1 is an empirical social network that encodes mentoring relationships among corporate executives \cite{krackhardt1987cognitive}, and as such, our system  models a collective business decision in which each executive boardroom member has the unique support of a personalized AI agent. Layer 2 is created as a random directed graph. 
%
We will show for this system that the existence of a cooperative versus non-cooperative optimum  depends crucially on the relative   timescales of two coupled consensus systems. We again focus on   the case of  $T=2$ layers with Laplacians given by $  \chi \mathbf{L}^{(1)}$ and $ (1-\chi) \mathbf{L}^{(2)} $, where  \emph{rate-scaling parameter} $\chi\in(0,1)$ tunes the relative convergence rate for each layer. 
We insert  these weighted Laplacians into \eqref{eq:L_bar_2} to obtain 
$\overline{\mathbf{L}}(\delta) = \frac{1 + \delta}{2}  \chi\mathbf{L}^{(1)} + \frac{1 - \delta}{2}  (1-\chi)\mathbf{L}^{(2)},$ 
and then study how a system's behavior (i)--(v) depends on both $\chi$ and $\delta$.
Note that the introduction of $\chi$ changes the layers' separate convergence rates to be $\chi \text{Re}(\lambda_2^{(1)})$ for layer 1 and $(1-\chi) \text{Re}(\lambda_2^{(2)})$ for layer 2.

\subsection{A Cooperative Optimum Requires that Humans and AI-Agents Coordinate on  Similar Timescales}\label{sec:time_balance}

We now examine the influence of layers’ relative timescales on the convergence rate for the human-AI system shown in Fig.~\ref{fig1}(A). In Fig.~\ref{fig:timescales_HumanAI}, we present results for an experiment that is similar to the results shown in Fig.~\ref{fig:timescales_Simple}.
We plot the convergence rate $\text{Re}(\lambda_2)$ versus $\delta$ for the human-AI system, and different columns reflect different choices for $\chi$. In each panel, different curves reflect different choices for $\omega$, and black curves indicate our analytical prediction for large $\omega$. By comparing across the columns, observe that their optima are cooperative for intermediate values of $\chi$ (i.e., the second, third and fourth columns) and  non-cooperative when $\chi$ is either too small or large (e.g., see left-most and right-most columns).  That is, a cooperative optimum requires that the humans and AI-agents coordinate on similar timescales. Otherwise, we find that the convergence is fastest when the faster layer influences the slower one without a reciprocated influence, i.e., layer dominance.

\subsection{Cooperative Optimum Yields Fastest Convergence}\label{sec:overall}

\begin{figure}[!t]
    	\centering
    	\includegraphics[width=\linewidth]{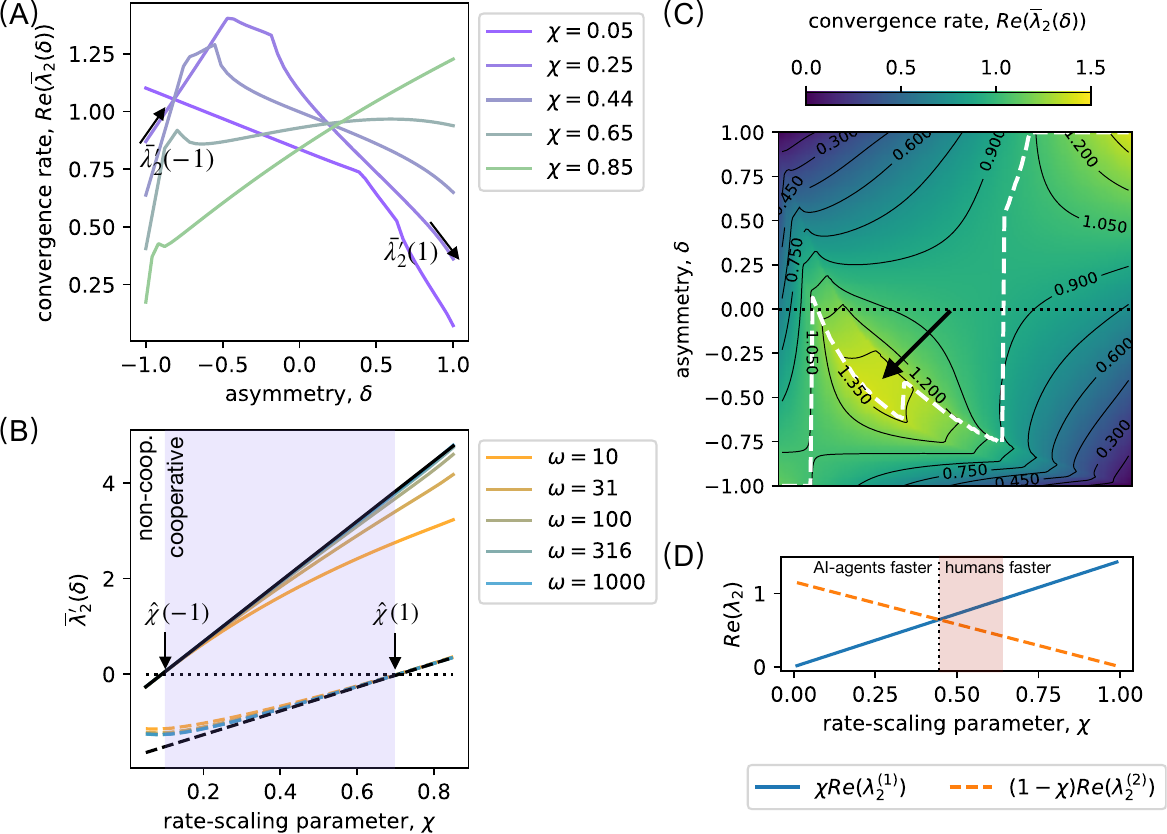}
    	\caption{ {\bf Combined effects of coupling asymmetry and relative timescale on a Human-AI consensus system}.
	We present identical information as in Fig.~\ref{fig:optima_Simple}
 but for the Human-AI network that was shown in Fig.~\ref{fig1}(A).
	The main difference is highlighted by the arrow in panel (C): the overall optimum $\max_{\delta,\chi} \text{Re}\left(\overline{\lambda}_2(\delta)\right)$ is now   a cooperative optimum that occurs at approximately $ ( {\delta}, {\chi})  \approx(-0.5,0.25)$.
	That is, if one allows coordination within the social-network layer to be faster than that of the AI-agents, or vice versa, and one also allows for asymmetric coupling between these two consensus systems,
	 then for this particular multiplex network, the overall system converges optimally fast   when the AI layer   coordinates  slightly faster than, and has a slightly stronger influence over, the social-network layer.}
    	\label{fig:optima_HumanAI}
\end{figure}

In Fig.~\ref{fig:optima_HumanAI}, we present results  that are identical to those that were presented in Fig.~\ref{fig:optima_Simple}, except we now examine the convergence rate for the human-AI system shown in Fig.~\ref{fig1}(A). In Fig.~\ref{fig:optima_HumanAI}(A), we plot $\text{Re}\left(\overline{\lambda}_2(\delta)\right)$ versus $\delta$ for several choices of $\chi$. 
Observe that the optimum is cooperative for $\chi \in\{ 0.05,0.85\}$  and  non-cooperative for $\chi\in \{0.25,0.45,0.65\}$. This is in agreement with our theory, which is shown Fig.~\ref{fig:optima_HumanAI}(B). Specifically, we find this system to exhibit a cooperative optimum for the approximate range $\chi \in(0.1,0.7)$ and a non-cooperative optimum outside this range. That is, the rates of AI-AI  coordination/communicate should be sufficiently similar to that for human-human coordination/communicate, otherwise the optimally coupled system will be non-cooperative---i.e., convergence will be optimally fast when one network layer influences the other without feedback. 

In Fig.~\ref{fig:optima_HumanAI}(C), we plot $\text{Re}\left(\overline{\lambda}_2(\delta)\right)$ across the parameter space $\delta\in[-1,1]$ and $\chi\in[0,1]$. The dashed white  curve in  indicates the optimum $\hat{\delta} = \text{argmax}_\delta \text{Re}(\overline{\lambda}_2)$ for each value of $\chi$. The black arrow highlights that the overall optimum is cooperative and occurs  at the approximate location    $( {\delta}, {\chi})  \approx(-0.5,0.25)$. Intuitively, consensus is maximally accelerated  when the AI-agents   are slightly more influential (i.e., since $ {\delta}<0$), and collective dynamics among humans are slightly slowed down (i.e., since $ {\chi}<0.5$).

Finally, in Fig.~\ref{fig:optima_HumanAI}(D),  we highlight that an unintuitive property also occurs for this humam-AI system.
Specifically, 
the shaded area highlights for the approximate range $\chi\in(0.45,0.65)$ that convergence is fastest if the slower layer (the AI-agents layer) is actually more influential than the faster layer (the human layer). 

Before concluding, we emphasize that these specific findings  are a result of the particular network layers that we study (i.e., layer 1 is an empirical  social network \cite{krackhardt1987cognitive} and we generate layer 2 as a random directed graph). Thus, these network-specific findings describe the behavior of this specific  model and should not be simply extrapolated to real-world  decision systems. In general, the optimal coupling of a human-AI decisions system will greatly differ from one application to another depending on the specific details of each system, which includes the layers' unique network topologies, their respective dynamics, and one's design goals for that systems. In principle, one should also consider other design factors beyond convergence rate \cite{ezer2019trust,bansal2019beyond}. Nevertheless, our theory provides a baseline of understanding for the effects on convergence rate   for human-AI decision systems in which these two network types (a social network and coordinating AI agents) are asymmetrically coupled and have different timescales for collective coordination.

\section{Discussion}\label{sec:Disc}

By now, the scientific literature on multiplex network dynamics is well established \cite{cozzo2018multiplex}. However, most   theory focuses on network layers that are symmetrically coupled, and  there remains a lack of understanding of how the asymmetric coupling of layers can affect dynamics and also provide new strategies for optimization. This is troubling since real-world networks are often interconnected with one layer being more influential than another. Thus motivated, we formulated a model for multiplex networks in which  an asymmetry parameter $\delta$ can tune the extent to which interlayer influences are biased in a particular direction. Although our work is primarily motivated by modeling collective behavior over a human-AI   system, our formulation of a supraLaplacian $\mathbb{L}(\omega,\delta)$  in \eqref{eq:supra} using an asymmetry parameter $\delta\in[-1,1]$ can support the broader study of how coupling asymmetry can affect any Laplacian-related dynamics (i.e., diffusion, synchronization, and so on).

Here, we have focused on the impact of coupling asymmetry on  the convergence rate $\text{Re}(\lambda_2)$ toward a collective state, which is an important property that is often optimized for engineered systems \cite{assran2019stochastic,niwa2020edge,vogels2020powergossip,kong2021consensus,huynh2021}. 
We provided an initial observation and categorization revealing five distinct ways [see (i)--(v) in Sec.~\ref{sec:asaaq}] in which coupling asymmetry has a nonlinear effect on  $\text{Re}(\lambda_2)$. Moreover, it is insightful to consider systems that are optimally coupled and ask whether their coupling is \emph{cooperative}, in that they mutually influence each other, or \emph{non-cooperative}, in that one system directs another without reciprocated feedback. It's worth highlighting that the situation of  non-cooperative coupling between network layers closely relates to prior research for collective dynamics over `master-slave' systems \cite{haken1977synergetics,suykens1997master,ramirez2018enhancing}. These similarities contain subsystems in which one influences another without reciprocated feedback. Our work extends their study  to the setting of optimized multiplex networks.

We find a non-cooperative configuration to be optimal when one system is much faster than the other, whereas a cooperative coupling is optimal if the layers have sufficiently similar dynamics timescales. In fact, we obtained a theoretical criterion (see Sec.~\ref{sec:exist}) for when the fastest convergence is   {cooperative} versus non-{cooperative}.  This result, in addition to the system properties (i)--(v) and the structural/dynamical factors influencing (non)-cooperation for optimal systems, should be considered as a stepping stone for further research  on the optimization of interconnected systems using techniques that jointly consider coupling asymmetry and timescale tuning. It would be interesting to explore whether our findings/methods are also predictive for more-complicated dynamics including, e.g., the optimization of synchronized chaotic systems \cite{pecora1998master,sun2009master,skardal2017optimal} 
and empirical human-AI systems.

See \cite{zhao_code} for a codebase that models interconnected consensus systems and reproduces our findings.


%


\ifCLASSOPTIONcompsoc
  \section*{Acknowledgments}
\else
  \section*{Acknowledgment}
\fi

The authors would like to thank Sarah Muldoon, Naoki Masuda, and Malbor Asllani for helpful feedback.

\ifCLASSOPTIONcaptionsoff
  \newpage
\fi



\bibliographystyle{IEEEtran}
\bibliography{bibliography}
%
%
%


\begin{IEEEbiography}[{\includegraphics[width=1in,height=1.25in,clip,keepaspectratio]{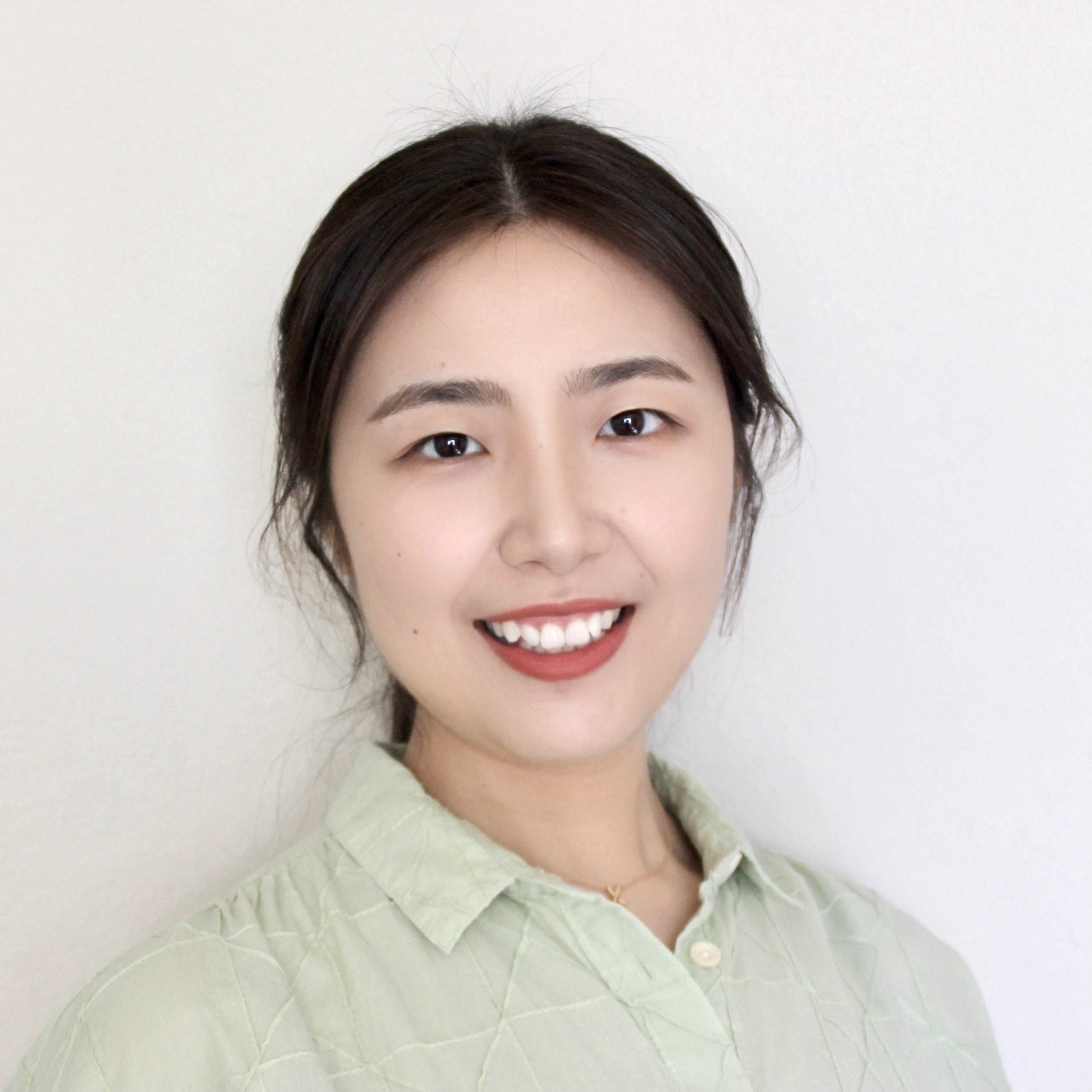}}]{Zhao Song}
received a B.S. degree in Information and Computing Science from Henan University of Technology, China, in 2015. and the M.A. and Ph.D. degrees in Mathematics from the University at Buffalo, the State University of New York, USA, in 2018 and 2021. She is currently working as a research associate at the Department of Mathematics, Dartmouth College, USA. Her research focuses on multiplex networks and network dynamics in network science.
\end{IEEEbiography}

\begin{IEEEbiography}[{\includegraphics[width=1in,height=1.25in,clip,keepaspectratio]{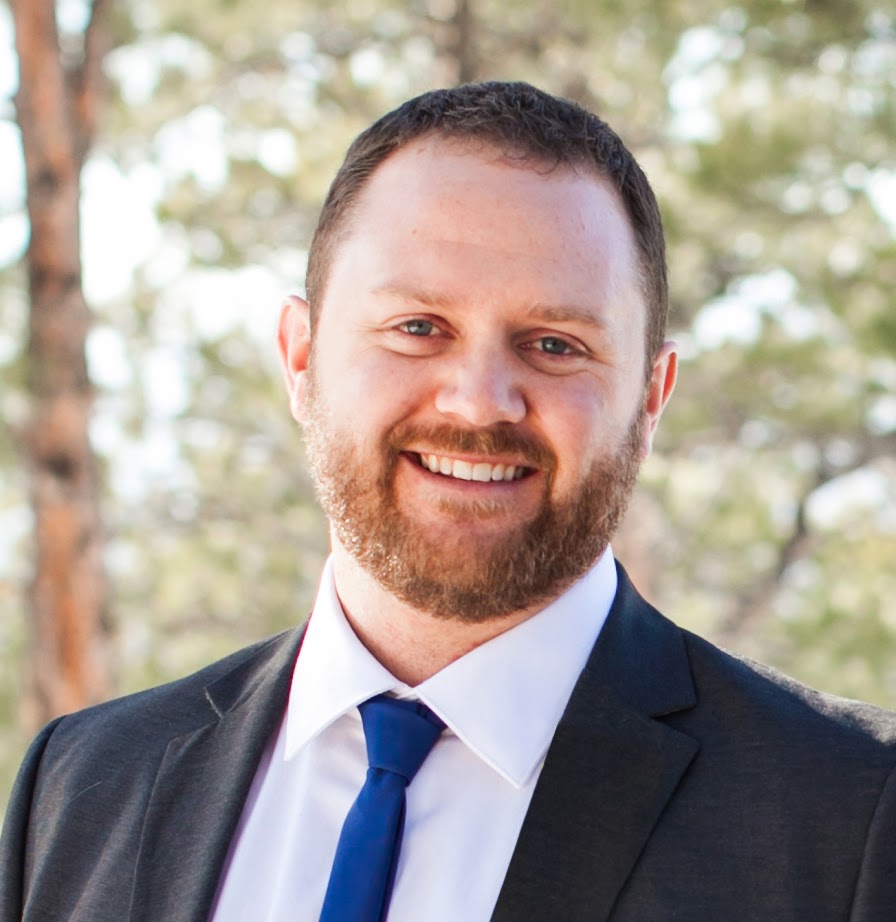}}]{Dane Taylor}
received a PhD in applied mathematics from the University of Colorado, Boulder in 2013. Before joining the University of Wyoming in 2023, Dr. Taylor held   positions at the Statistical and Applied Mathematical Sciences Institute (2013-2015), the University of North Carolina at Chapel Hill (2015-2017), and the University at Buffalo, SUNY (2017-2023).
Dr. Taylor has published over 35 papers on   network  modeling for data science and dynamical systems and has been the main organizer for events including the Northeast Regional Conference on Complex Systems (NERCCS) and the SIAM Workshop on Network Science.
\end{IEEEbiography}

%


%


\vfill



%

\newpage

\appendices
\renewcommand{\theequation}{\thesection.\arabic{equation}}
\setcounter{page}{1}

\section{Further Study of Nonrobust Optima}\label{sec:nonr}

\begin{figure}[!b]
	\centering
    \includegraphics[width=.9\linewidth]{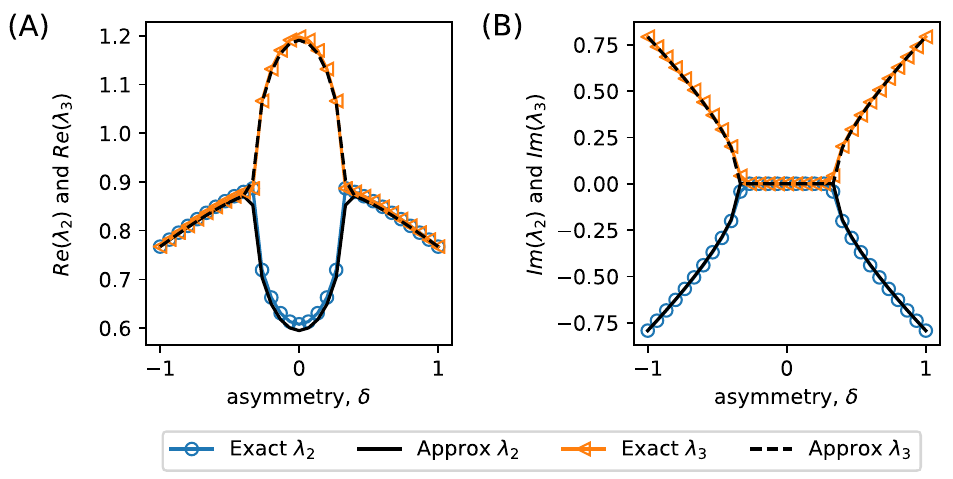}
	\caption{{\bf A collision between eigenvalues $\lambda_2$ and $\lambda_3$ yields a spectral bifurcation and nonrobust optima.}
	We plot (A) the real parts and (B) imaginary parts of $\lambda_2$ and $\lambda_3$ versus $\delta$ for the interconnected consensus systems shown in the center column of Fig.~2(A). Symbols depict    observed eigenvalues that are directly computed using $\mathbb{L}(\omega,\delta)$ with $\omega=30$, whereas   black curves are theoretical predictions.
	The nonrobust optima at $\delta=\pm 0.35$ coincide with a spectral bifurcation in which $\lambda_2$ and $\lambda_3$ collide to yield a complex pair of eigenvalues.
	}
	\label{fig:bifurcation}
\end{figure}

Here, we further investigate $\text{Re}(\lambda_2)$ for the multiplex network shown in  the center column of Fig.~2(A). In the center column of Fig.~2(B), we plotted $\text{Re}(\lambda_2)$ versus $\delta$ for these interconnected consensus systems and observed two optima near $\delta=\pm0.35$. We  classify these as ``nonrobust'' optima, because the derivative $\frac{d}{d\delta} \text{Re}(\lambda_2)$ is undefined (i.e., discontinuous) at these optima. 

Here, we show that these two nonrobust optima (and the lack of differentiability) arise due to a spectral bifurcation  in which the two eigenvalues  $\lambda_2$ and $\lambda_3$ of $\mathbb{L}(\omega,\delta)$ collide and give rise to a complex pair of eigenvalues. Here, we have defined $\lambda_2$ and  $\lambda_3$ as the eigenvalues of $\mathbb{L}(\omega,\delta)$  with second-smallest and third-smallest real part, respectively.

In Fig.~\ref{fig:bifurcation}, we plot $\lambda_2$ and $\lambda_3$ for a supraLaplacian $\mathbb{L}(\omega,\delta)$ associated with with $\omega=30$ and various $\delta$ for the system shown in the center column of Fig.~2(A). In panels (A) and (B) of Fig.~\ref{fig:bifurcation}, we depict the real and imaginary parts, respectively, of these eigenvalues. Symbols indicate observed values that are   directly computed  using $\mathbb{L}(\omega,\delta)$, whereas  solid and dashed black curves  indicate  theoretical predictions that we will present in the next section.
Note that the observed and predicted eigenvalues are in excellent agreement.

Observe in Fig.~\ref{fig:bifurcation}(A) that the theoretical curve for $\lambda_2$ is identical to the one  that is in the center column of Fig.~2(B). The curve has two optima near $\delta\pm 0.35$, and these values coincide with spectral bifurcations in which $\lambda_2$ and $\lambda_3$ change from being distinct real-valued eigenvalues to being a complex pair of eigenvalues (or vice versa). That is, for $|\delta|\le0.35$, the eigenvalues $\lambda_2$ and $\lambda_3$ are purely real (i.e., $\text{Im}(\lambda_2)=\text{Im}(\lambda_3)=0$ and $\text{Re}(\lambda_3)>\text{Re}(\lambda_2)$. In contrast,  for $|\delta|>0.35$, the eigenvalues $\lambda_2$ and $\lambda_3$ are complex numbers and $\text{Re}(\lambda_3)=\text{Re}(\lambda_2)$. Thus, the nonrobust optima at $\delta =\pm 0.35$ arise for these interconnected consensus systems because of a spectral bifurcation.

\begin{figure}[!b]
    	\centering
    	\includegraphics[width=\linewidth]{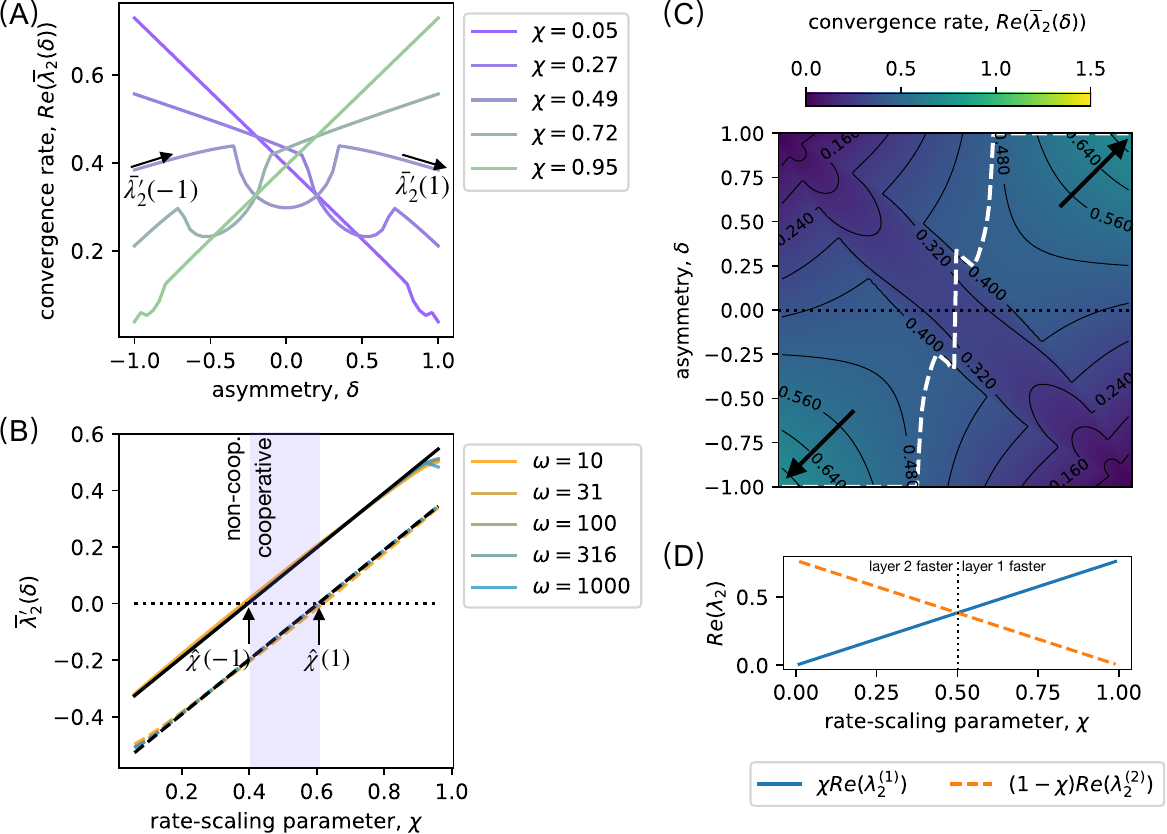}
    	\caption{ 
	{\bf Combined effects of coupling asymmetry and  relative timescale on the third system in Fig.~\ref{fig:behaviors_toy}(A).}
	 We present identical information as in Figs. 4 and 5 but for the multiplex network that was visualized in the third column of Fig.~\ref{fig:behaviors_toy}. The main difference here can be seen in panels (C) and (D). The two arrows in (C) highlight that there are now two overall optimal choices for $\delta$ and $\chi$ (both are non-cooperative), and the lack of a shaded region in (D) highlights that there is no counter-intuitive phenomenon whereby it is optimal for the slower system to more strongly influence the faster one.}
    	\label{fig:optima_RingGraph7}

\end{figure}

In Fig.~\ref{fig:optima_RingGraph7}, we present a information that is identical to what we presented in   Figs.~\ref{fig:optima_Simple}
 and \ref{fig:optima_HumanAI}
, except that we now consider the multiplex network shown in the center column of Fig.~2(A). In Fig.~\ref{fig:optima_RingGraph7}
(A), we plot $\text{Re}\left(\overline{\lambda}_2(\delta)\right)$ versus $\delta$ for several choices of $\chi$. 
Observe for $\chi = 0.49$ that there exist two nonrobust optima because of spectral bifurcations, which we explained in Fig.~\ref{fig:bifurcation}.
This is the only value of $\chi$ that we examine for which our theory guarantees  cooperative optimum, and in fact, a non-cooperative optimum can be observed for the other values $\chi\in\{0.05,0.27,0.72,0.95\}$.
Specifically,  we depict our theory from Section~\ref{sec:exist} in Fig.~\ref{fig:optima_RingGraph7}
(B), and the shaded  $\chi\in(0.39,0.6)$ indicates the values of $\chi$ for which a cooperative optimum is guaranteed. 

In Fig.~\ref{fig:optima_RingGraph7}
(C), we visualize $\text{Re}\left(\overline{\lambda}_2(\delta)\right)$ for the parameter space $\delta\in[-1,1]$ and $\chi\in[0,1]$. As before, the dashed white  curve   indicates the optimum $\hat{\delta} = \text{argmax}_\delta \text{Re}(\overline{\lambda}_2)$ asymmetry for each value of $\chi$. The black arrow highlights that the overall optimum occurs approximately when either $( {\delta}, {\chi})  =  (-1,0)$ or  $( {\delta}, {\chi})  =  (1,1)$. That is, the overall optimum only occurs when one layer is set to be as fast as possible and it non-cooperatively influences the other layer. Due to symmetry, it doesn't matter which layer is chosen to be the dominator.

In Fig.~\ref{fig:optima_RingGraph7}
(D),   solid blue and dashed orange lines in  indicate the separate convergence rate of each (uncoupled) system:  $\chi \text{Re}(\lambda_2^{(1)})$ for system 1 and $(1-\chi) \text{Re}(\lambda_2^{(2)})$ for system 2.  Note for   $\chi >0.5$ that layer 1 is the faster system, and vice versa for $\chi<0.5$. Compare this to the dashed white  curve in Fig.~\ref{fig:optima_RingGraph7}
(C), where one can observe that $\hat{\delta}<0$ for $\chi< 0.5$ and $\hat{\delta}>0$ for $\chi> 0.5$. Therefore, for this system, it is always optimal for the faster layer to more strongly influence the slower layer. See Figs.~\ref{fig:optima_Simple}
(D) and \ref{fig:optima_HumanAI}
(D) for counter-intuitive situations in which the convergence is fastest  when the slower system more-strongly influences the faster one.

\section{Derivation of Equations~\eqref{eq:lim_lam2}--\eqref{eq:L_bar_2}}\label{sec:appendix_A}

\setcounter{equation}{0} 
Here, we provide our derivation of \eqref{eq:lim_lam2} and \eqref{eq:L_bar}, which predict the eigenvalue $\lambda_2(\omega, \delta)$ of $\mathbb{L}(\omega, \delta)$ in the limit of large $\omega$. We do not have $\omega$ and $\delta$ in the main paper, but will keep these parameters in the appendix. We will use perturbation theory for directed multiplex networks that is similar to that which was developed in  \cite{taylor2019tunable,taylor2020multiplex}.

First, we introduce a change of variables   $\epsilon = 1/\omega$ and  multiply   both sides of \eqref{eq:supra} 
by $\epsilon$ to obtain
\begin{equation} \label{eq:supra_tilde}
    \mathbb{\tilde{L}}(\epsilon, \delta) = \epsilon\mathbb{L}(\epsilon^{-1}, \delta) = \epsilon\mathbb{L}^{ \mathrm{I}} + \mathbb{L}^{\mathrm{I}}(\delta)\,  .
\end{equation}
Because we've only scaled the matrix $\mathbb{L}(\omega, \delta)$ by $\epsilon$, it follows that $\tilde{\lambda}_2(\epsilon, \delta)=\epsilon {\lambda}_2$ is an eigenvalue of $\mathbb{\tilde{L}}(\epsilon, \delta)$. Let $\mathbbm{\ {u}}(\epsilon, \delta)$ and $\mathbbm{\ {v}}(\epsilon, \delta)$ be the associated left and right eigenvectors of  $\mathbb{\tilde{L}}(\epsilon, \delta)$.
Note that   scalar multiplication does not change the eigenvectors of a matrix, and so $\mathbbm{\ {u}}(\epsilon, \delta)$ and $\mathbbm{\ {v}}(\epsilon, \delta)$ are also the eigenvectors of $\mathbb{L}(\omega, \delta)$ that are associated with $\lambda_2$.
The main motivation for this transformation is that we can more easily study
${\lambda}_2$ in $\omega\to\infty$ limit by instead studying $\tilde{\lambda}_2(\epsilon, \delta)$ as $\epsilon\to0^+$.
%
In this limit,
\begin{equation}
\mathbb{\tilde{L}}(\epsilon, \delta) \rightarrow \mathbb{L}^{ \mathrm{I}}(\delta) = \bm{ L}^{\mathrm{I}}(\delta) \otimes \mathrm{I} \, .
\end{equation}

To proceed, we   first establish some properties about the eigenvalues and eigenvectors of $\bm{L}^{\mathrm{I}}(\delta) \otimes \mathrm{I}$.

\begin{lemma}\label{thm:thm1}
Let $\mu$ be an eigenvalue of $\bm{L}^{\mathrm{I}}(\delta) \in\mathbb{R}^{T\times T}$ and  $\bm{u}$ and $\bm{v}$, respectively,  be its associated left and right eigenvectors. Furthermore, let $\mathbf{e}^{(i)} \in\mathbb{R}^N $ denote the $i$-th unit vector such that all entries are zeros, expect for entry $i$, which is a one.  It then follows that $\mu$ is  an eigenvalue of $\bm{ L}^{\mathrm{I}}(\delta) \otimes \mathrm{I}$, and it has associated left and right eigenvectors given by
\begin{align} \label{eq:u^it}
    \tilde{\mathbbm{u}}^{(i)} = \mathbb{P} \left(\mathbf{e}^{(i)}\otimes \bm{u}\right), ~~~
    \tilde{\mathbbm{v}}^{(i)} = \mathbb{P}\left(\mathbf{e}^{(i)}\otimes \bm{v}\right) \, ,
\end{align}  
where   $\mathbb{P}$ is  a ``stride permutation matrix''  that contains entries
\begin{align}
	P_{ij} = 
	\left\{ \begin{array} {rl} 
	1, & j = \lceil i/N\rceil + T[(i - 1)\mod N]\\
	0, & \text{otherwise} \, .
	\end{array} \right.  \nonumber
\end{align}
\end{lemma}

\begin{remark}
As discussed in \cite{taylor2017eigenvector,taylor2019tunable}, this stride permutation is a unitary matrix that changes the ordering of indices for a supramatrix associated with a multiplex network. That is,  the indices originally count by nodes and then layers, but after applying $\mathbb{P}$, the counting is first by layers and then by nodes.
\end{remark}

\begin{remark}
    Lemma~\ref{thm:thm1} is true for any choice $i\in\{1,\dots,N\}$, and so each eigenvalue $\mu$ of $\bm{L}^\mathrm{I}(\delta) \otimes \mathrm{I}$ has an eigenspace that is at least $N$-dimensional . [It may be larger if $\mu$ is a repeated eigenvalue of $\bm{L}^\mathrm{I}(\delta)$.]
\end{remark}

\begin{proof}
    The stride permutation yields an identity
    \begin{align}
        \bm{L}^{\mathrm{I}}(\delta) \otimes \mathrm{I}
        = \mathbb{P}  \left( \mathrm{I} \otimes  \bm{L}^{\mathrm{I}}(\delta) \right) \mathbb{P}^{\mathrm{T}} \, ,
    \end{align}
    where
    \begin{align}
        \mathrm{I} \otimes  \bm{L}^{\mathrm{I}}(\delta) = 
         \left(\begin{matrix}
         \bm{L}^{\mathrm{I}}(\delta) & 0 & \cdots\ & 0 \\
        0 &  \bm{L}^{\mathrm{I}} (\delta) & \cdots\ & 0 \\
        \vdots & \vdots & \ddots & \vdots\\
        0 & 0 & \cdots\ &  \bm{L}^{\mathrm{I}}(\delta)
        \end{matrix}\right) \, .
    \end{align}
    Since $\mathbb{P}$ is a unitary matrix, the eigenvalues of $ \bm{L}^{\mathrm{I}}(\delta) \otimes {I}$ are identical to those of  $ \mathrm{I} \otimes  \bm{L}^{\mathrm{I}}(\delta)$. Moreover, if $\mathbbm{v}$ is a right eigenvector of $\mathrm{I} \otimes  \bm{L}^{\mathrm{I}}(\delta) $ with eigenvalue $\mu_2$, then $\mathbb{P}\mathbbm{v}$ is a right eigenvector of $\bm{L}^{\mathrm{I}}(\delta) \otimes \mathrm{I}$. This can easily be checked:
    \begin{align}
        \left( \bm{L}^{\mathrm{I}}(\delta) \otimes \mathrm{I} \right) \mathbb{P}\mathbbm{v} 
        &= \mathbb{P}  \left( \mathrm{I} \otimes  \bm{L}^{\mathrm{I}}(\delta)\right) \mathbb{P}^{\mathrm{T}} \mathbb{P}\mathbbm{v} \nonumber\\
        &= \mathbb{P}  \left( \mathrm{I} \otimes  \bm{L}^{\mathrm{I}}(\delta) \right) \mathbbm{v} \nonumber\\
        &= \mu \mathbb{P}   \mathbbm{v}.
    \end{align}
    One can similarly show that if $\mathbbm{u}$ is a left eigenvector for  $\mathrm{I} \otimes  \bm{L}^{\mathrm{I}}(\delta) $, then $\mathbb{P}\mathbbm{u}$ is one for $\bm{L}^{\mathrm{I}}(\delta) \otimes \mathrm{I}$.
    What remains for us to show is that $\mathbbm{u}=\mathbf{e}^{(i)}\otimes \bm{u}$ and $\mathbbm{v}=\mathbf{e}^{(i)}\otimes \bm{v}$ are left and right eigenvectors of $\mathrm{I} \otimes  \bm{L}^{\mathrm{I}}(\delta) $.
    We prove this using a standard property for the product of two Kronecker products:
    \begin{align}
        \left(\mathrm{I}  \otimes  \bm{L}^{\mathrm{I}}(\delta) \right)\left( \mathbf{e}^{(i)}\otimes \bm{v} \right) 
        &=   \left(\mathrm{I} ~\mathbf{e}^{(i)}\right) \otimes \left(\bm{L}^{\mathrm{I}}(\delta) \bm{v} \right)  \nonumber\\
        &= \mathbf{e}^{(i)} \otimes \mu   \bm{v}  \nonumber\\
        &= \mu(\mathbf{e}^{(i)} \otimes    \bm{v}).
    \end{align}
    A similar result can be obtained for the left eigenvector $\mathbf{e}^{(i)}\otimes \bm{u}$. 
\end{proof}

Having established basic results for the spectral properties of $\bm{L}^\mathrm{I}(\delta) \otimes \mathrm{I}$ in Lemma~\ref{thm:thm1}, we are now ready to study the eigenvalue $\tilde{\lambda}_2(\epsilon,\delta)$ of $\tilde{\mathbb{L}}(\epsilon, \delta)$ that has the second-smallest real part in the limit $\epsilon\to0^+$. 
We formalize this result with in following theorem.

\begin{theorem}
    Assume that a supraLaplacian $\mathbb{L}(\omega,\delta)$ 
    corresponds to a strongly connected graph.
    Further, let $\bm{u}(\delta)=[u_1(\delta),\dots,u_T(\delta)]^\mathrm{T}$ be the left eigenvector of an interlayer Laplacian $\bm{L}^{\mathrm{I}}(\delta)$ that is associated with the zero-valued (i.e., trivial) eigenvalue.
    Then the eigenvalue $\lambda_2(\omega,\delta)$ of  $\mathbb{L}(\omega,\delta)$ that has second-smallest real part has the following limit: 
    \begin{align}\label{eq:lim_lam244}
        \lim_{\omega\to\infty} \lambda_2(\omega,\delta)  &= \overline{\lambda}_2(\delta)\, ,
    \end{align}
    where $\overline{\lambda}_2(\delta)$ is the  eigenvalue  of  matrix 
    \begin{eqnarray}\label{eq:L_bar44}
        \overline{\mathbf{L}}(\delta)
        &=&   \sum_{t=1}^\mathrm{T} w_{t}(\delta)    \mathbf{L}^{(t)}
    \end{eqnarray}
    that has the second-smallest real part, and $w_t (\delta) = u_t(\delta)/ {\sum_{t'} u_{t'}(\delta)}$.

\begin{proof}
    Consider first-order Taylor expansions for the eigenvalue $\tilde{\lambda}_2(\epsilon,\delta)$ of $\mathbb{\tilde{L}}(\epsilon, \delta)$ and its associated left and right eigenvectors:
    \begin{subequations}
    \begin{eqnarray} \label{eq:perturbation}
    \tilde{\lambda}_2(\epsilon, \delta) 
    &=& \tilde{\lambda}_2(0, \delta) 
    + \epsilon \tilde{\lambda}^{\prime}_2(0, \delta) 
    + \mathcal{O}(\epsilon^2)  ,  \\
    \tilde{\mathbbm{u}}(\epsilon, \delta) &=& \tilde{\mathbbm{u}} (0, \delta) + \epsilon{\tilde{\mathbbm{u}}'}(0, \delta) + \mathcal{O}(\epsilon^2),  \\
    \tilde{\mathbbm{v}}(\epsilon, \delta) &=& \tilde{\mathbbm{v}} (0, \delta) + \epsilon{\tilde{\mathbbm{v}}'}(0, \delta) + \mathcal{O}(\epsilon^2)  . 
    \end{eqnarray}
    \end{subequations}
    Note that we have   defined the   derivatives
    \begin{subequations}
    \begin{align}
    	\tilde{\lambda}_2'(\epsilon,\delta) \equiv&    \frac{d}{d\epsilon} \tilde{\lambda}_2(\epsilon,\delta),  \\
    	{\tilde{\mathbbm{u}}'}(\epsilon,\delta) \equiv& \frac{d}{d\epsilon} {\tilde{\mathbbm{u}}}(\epsilon,\delta),  \\
    	{\tilde{\mathbbm{v}}'}(\epsilon,\delta) \equiv& \frac{d}{d\epsilon} {\tilde{\mathbbm{v}}}(\epsilon,\delta)  .
    \end{align}
    \end{subequations}
    
    We first consider the term $\tilde{\lambda}_2(0, \delta) $.
    Since we assumed $\bm{L}^{\mathrm{I}}(\delta)$ to be the Laplacian of a  strongly connected graph, the smallest eigenvalue $\mu_1=0$ of $\bm{L}^{\mathrm{I}}(\delta)$ is guaranteed to be a simple eigenvalue with multiplicity 1.
    Lemma~\ref{thm:thm1} implies $\mu_1=0$ is an eigenvalue of $\bm{L}^{\mathrm{I}}(\delta) \otimes \mathrm{I}$ with an $N$-dimensional eigenspace. Hence,  matrix $\tilde{\mathbb{L}}(\epsilon,\delta)$ has $N$ eigenvalues that converge to $0$ as $\epsilon\to0^+$ (and in fact, one of these eigenvalues is always exactly equal to zero). By   definition (i.e., since it has the smallest, positive real part), the eigenvalue $\tilde{\lambda}_2(\epsilon,\delta)$ of $\tilde{\mathbb{L}}(\epsilon,\delta)$ must be one of these eigenvalues, which implies that  $\tilde{\lambda}_2(0,\delta)=0$.
    
Next, we consider the derivative term $\tilde{\lambda}^{\prime}_2(0, \delta)  $. We will show that it equals the second smallest eigenvalue of the  matrix defined in \eqref{eq:L_bar}.
    To this end, we consider the eigenvalue equation 
    \begin{align}
    \tilde{\mathbb{L}}(\epsilon,\delta)\tilde{\mathbbm{v}}(\epsilon, \delta) 
    &= \tilde{\lambda}_2(\epsilon,\delta)
    \tilde{\mathbbm{v}}(\epsilon, \delta) ,
    \end{align}
    and we expand all terms to first order to obtain
    \begin{align}\label{eq:expand}
        \Big[\epsilon\mathbb{L}^{\mathrm{L}}&+\mathbb{L}^{\mathrm{I}}(\delta)\Big] 
        \Big[\tilde{\mathbbm{v}}(0, \delta) + \epsilon{\tilde{\mathbbm{v}}}^{\prime}(0, \delta) \Big] \nonumber\\
        = & \Big[\tilde{\lambda}_2(0, \delta) + \epsilon \tilde{\lambda}_2'(0, \delta) \Big]
        \Big[\tilde{\mathbbm{v}}(0, \delta) + \epsilon {\tilde{\mathbbm{v}}}^{\prime}(0, \delta) \Big]. 
    \end{align}
    The zeroth-order and first-order terms must be consistent, which gives rise to two separate equations. The  equation associated with the zeroth-order terms yields an eigenvalue equation 
    \begin{align}
        \mathbb{L}^{\mathrm{I}} (\delta) \tilde{\mathbbm{v}} (0, \delta)
        &= \tilde{\lambda}_2(0, \delta)   \tilde{\mathbbm{v}} (0, \delta).
    \end{align}
    Lemma~\ref{thm:thm1} implies that the eigenvalue $\tilde{\lambda}_2(0, \delta) =0$ has an $N$-dimensional right eigenspace spanned by   right eigenvectors having the form $\tilde{\mathbbm{v}}^{(i)} = \mathbb{P}\left(\mathbf{e}^{(i)}\otimes \bm{v}\right) $  for $i\in\{1,2,\dots,N\}$. We similarly define $\tilde{\mathbbm{u}}^{(i)} = \mathbb{P}\left(\mathbf{e}^{(i)}\otimes \bm{u}\right) $ for the left eigenspace. Note that the vectors $\bm{u}$ and $\bm{v}$ contain entries that are nonnegative, which can be proved using the Perron--Frobenius theorem.
    
    Next, we expand $\tilde{\mathbbm{v}} (0, \delta)$ in this eigenbasis as
    \begin{align}\label{eq:general}
        \tilde{\mathbbm{v}} (0, \delta) = \sum_i \tilde{\alpha}_i \tilde{\mathbbm{v}}^{(i)}.
    \end{align}
    We  define a vector of  coordinates $\tilde{\bm{\alpha}} =[\tilde{\alpha}_1,\dots, \tilde{\alpha}_N]^{\mathrm{T}}$ that must be determined, and we note that the vector must be normalized with $||\tilde{\bm{\alpha}}||_2=1$.
    
    The first-order terms in \eqref{eq:expand} give rise to a linear   equation
    \begin{align}
        \mathbb{L}^{\mathrm{L}}  \tilde{\mathbbm{v}} (0, \delta)
        + \mathbb{L}^{\mathrm{I}}(\delta) {\tilde{\mathbbm{v}} }^{\prime}(0, \delta)
        &= \tilde{\lambda}^{\prime}_2(0, \delta)   \tilde{\mathbbm{v}} (0, \delta) \, ,
    \end{align}
    which has used that $\tilde{\lambda}_2(0, \delta)=0$.    
    
    To solve for $\tilde{\lambda}^{\prime}_2(0, \delta) $, we left multiple by a left eigenvector 
    $\tilde{\mathbbm{u}}^{(i')} = \mathbb{P} \left(\mathbf{e}^{(i')}\otimes \bm{u}\right)$ of $\mathbb{L}^{\mathrm{I}}(\delta)$
    and again use $\tilde{\lambda}_2(0, \delta)=0$ to obtain
    \begin{equation}
        \left[\tilde{\mathbbm{u}}^{(i')} \right]^{*} 
        \mathbb{L}^{\mathrm{L}}\tilde{\mathbbm{v}}(0, \delta) 
        = \tilde{\lambda}^{\prime}_2(0, \delta)  
        \left[\tilde{\mathbbm{u}}^{(i')} \right]^{*} \tilde{\mathbbm{v}}(0, \delta) ,
    \end{equation}
    where $ \left[\tilde{\mathbbm{u}}^{(i')} \right]^{*} $ denotes the conjugate transpose of vector $\tilde{\mathbbm{u}}^{(i')}$. 
    Using the general form of $\tilde{\mathbbm{v}}(0, \delta)$ from \eqref{eq:general}, we obtain
    \begin{align}
        \sum_{i=1}^N \tilde{\alpha}_{i} \left[ \tilde{\mathbbm{u}}^{(i^{\prime})}\right]^{*} \mathbb{L}^{\mathrm{L}}  \tilde{\mathbbm{v}}^{(i)} 
        &= \tilde{\lambda}^{\prime}_2(0, \delta)  \sum_{i=1}^N \tilde{\alpha}_{i} \left[\tilde{\mathbbm{u}}^{(i^{\prime})}\right]^{*} \tilde{\mathbbm{v}}^{(i)}.
    \end{align}
    This system is identical to the following eigenvalue equation 
    $\overline{\mathbf{L}}(\delta)
    \tilde{\boldsymbol\alpha} = \tilde{\lambda}^{\prime}_2(0, \delta) \tilde{\boldsymbol\alpha} \,,
    $ 
    where
    \begin{align}
        \overline{\mathrm{L}}(\delta)_{i^{\prime}i}
        &= \frac{\left[\tilde{\mathbbm{u}}^{(i^{\prime})}\right]^{*}   \mathbb{L}^{\mathrm{L}}   \tilde{\mathbbm{v}}^{(i)}}{\left[\tilde{\mathbbm{u}}^{(i^{\prime})}\right]^{*} \tilde{\mathbbm{v}}^{(i)}} ,
    \end{align}
    $\tilde{\boldsymbol\alpha} = [\tilde{\alpha}_{1},\dots,\tilde{\alpha}_{n}]^{\mathrm{T}}$,
    and $\tilde{\lambda}_2'(0,\delta)=\overline{\lambda}_2(\delta)$ is the eigenvalue with second-smallest real part.
    We can further simplify this result using the definitions of $\tilde{\mathbbm{u}}^{(i')} = \mathbb{P} \left(\mathbf{e}^{(i')}\otimes \bm{u}\right)$ and $\tilde{\mathbbm{v}}^{(i)} = \mathbb{P} \left(\mathbf{e}^{(i)}\otimes \bm{v}\right)$ to obtain
    \begin{align}\label{eq:appp}
        \overline{\mathrm{L}}(\delta)_{i^{\prime}i}
        = \frac{\left[\bm{u}\right]^{*} \bm{L}^{(i^{\prime},i)} \bm{v}}
        {\left[\bm{u}\right]^{*} \bm{v}}
        = \sum_{t=1}^{T} 
        \left( \frac{ u_{t} }{\sum_t u_\tau} \right) \mathrm{L}_{i^{\prime}i}^{(t)} ,
    \end{align}
    where 
    \begin{align}\label{eq:appp2}
        \bm{L}^{(i^{\prime},i)} = \text{diag}\left(\mathrm{L}_{i^{\prime}i}^{(1)}, \mathrm{L}_{i^{\prime}i}^{(2)}, \dots, \mathrm{L}_{i^{\prime}i}^{(T)}\right),
    \end{align}
    and
    $\mathrm{L}_{i^{\prime}i}^{(t)}$ denotes the $(i^{\prime},i)$-component of   intralayer Laplacian $\mathbf{L}^{(t)}$.
    \eqref{eq:appp} follows after using     the definition of $\mathbb{P}$, which is a unitary matrix that permutes the enumeration of nodes and layers as described for Lemma~\ref{thm:thm1}. \eqref{eq:appp2}  follows after using that the right eigenvector $\bm{v}$ that is associated with the zero eigenvalue is spanned by the all-ones vector, $\bm{v} \varpropto[1,\dots,1]^{\mathrm{T}}$.
    Finally, we recall that  $\epsilon = 1/\omega$ and  
    $\epsilon \lambda_2(\omega,\delta) = \tilde{\lambda}_2(0, \delta) + \epsilon\tilde{\lambda}^{\prime}_2(0, \delta) + \mathcal{O}(\epsilon^2),
    $
    which implies
    \begin{align}
        \lim_{\omega\to\infty} \lambda_2(\omega,\delta) = \lim_{\epsilon\to0^+} \left[ 0 +  \tilde{\lambda}^{\prime}_2(0, \delta) + \mathcal{O}(\epsilon)\right] =  \overline{\lambda}_2(\delta).
    \end{align}
\end{proof}
\end{theorem}

\begin{corollary}
    Let $\mathbf{L}^{(1)}$ and $\mathbf{L}^{(2)}$ be the intralayer Laplacians of a two-layer multiplex network. The weighted-average Laplacian $\overline{\mathbf{L}}(\delta)$ given in \eqref{eq:L_bar} can be simplified as
    \begin{align} 
        \overline{\mathbf{L}}(\delta) = \left( \frac{1+\delta}{2} \right) \mathbf{L}^{(1)} + \left(\frac{1-\delta }{2}\right) \mathbf{L}^{(2)}.
    \end{align}

\begin{proof}
    The interlayer Laplacian  for $T=2$ asymmetrically coupled layers is given by \eqref{eq:LI_2}, and its zero-valued eigenvalue  has left eigenvector $\bm{u}(\delta)=[ 1+\delta ,1-\delta ]^\mathrm{T}$. The result follows after using that  $u_1(\delta) = 1+\delta$, $u_2(\delta) = 1-\delta$, and $u_1(\delta) +u_2(\delta) = 2$.
\end{proof}
\end{corollary}

\section{Derivation of Equations~\eqref{eq:lam2_prime} and \eqref{eq:critical_chi}}\label{sec:appendix_B}

\setcounter{equation}{0}
\begin{lemma}
Let $\mathbf{L}^{(1)}$ and $\mathbf{L}^{(2)}$ be the intralayer Laplacians of a two-layer multiplex network, and define $\chi\in(0,1]$ to be a time-scaling parameter that varies the relative timescale of dynamics for the two layers through the mapping:  $\mathbf{L}^{(1)}\mapsto \chi \mathbf{L}^{(1)}$  and $\mathbf{L}^{(2)}\mapsto (1-\chi) \mathbf{L}^{(2)}$. (Note that the dynamics of layer 1 is much faster as $\chi\to1$, whereas  layer 2 is much faster as $\chi\to0$.)
    Further, let $\overline{\mathbf{L}}(\delta)$ be the weighted-average  Laplacian  given in \eqref{eq:L_bar_2} under this mapping. It  then follows that  
    \begin{equation} \label{eq:L_bar_chi}
	 \overline{\mathbf{L}}(\delta) = \frac{1 + \delta}{2}\,\chi\mathbf{L}^{(1)} + \frac{1 - \delta}{2}\,(1-\chi)\mathbf{L}^{(2)} \,. 
    \end{equation}
\end{lemma}


\begin{theorem}\label{thm:t4}
Let $\overline{\lambda}_2'(\delta) \equiv \frac{d}{d\delta} \text{Re}\left(\overline{\lambda}_2(\delta) \right)$ be the derivative of the real part of 
$\overline{\lambda}_2(\delta)$ who is the second-smallest eigenvalue of the weighted-average  Laplacian $\overline{\mathbf{L}}(\delta)$. 
Further, let $\mathbf{u}^{(t)}$ and $\mathbf{v}^{(t)}$ be the left and right eigenvectors for the second smallest eigenvalue $\lambda^{(t)}_2$ of intralayer Laplacian $\mathbf{L}^{(t)}$ for $t\in\{1,2\}$, assuming $\lambda^{(t)}_2$ is simple.
We then find the following limits as $\delta \to \pm 1$
%
\begin{subequations}
\begin{align}\label{eq:lam2_prime44}
            \overline{\lambda}_2^{\prime}(1) &=
             \text{Re} \left( \frac{-{\mathbf{u}^{(1)}}^* \mathbf{L}^{(2)} \mathbf{v}^{(1)}} {2{\mathbf{u}^{(1)}}^* \mathbf{v}^{(1)}}  \right)\nonumber\\
            &~~~ + \chi \, 
             \text{Re} \left( \frac{{\mathbf{u}^{(1)}}^*
             \left(\mathbf{L}^{(1)} + \mathbf{L}^{(2)}\right) \mathbf{v}^{(1)}} 
             {2{\mathbf{u}^{(1)}}^* \mathbf{v}^{(1)}} \right),  \\
             \overline{\lambda}_2^{\prime}(-1)& =  
             \text{Re} \left( \frac{-\lambda^{(2)}_2}{2} \right) \nonumber\\
             &~~~ + \chi \,
             \text{Re} \left( \frac{{\mathbf{u}^{(2)}}^* 
             \left(\mathbf{L}^{(1)} +\mathbf{L}^{(2)} \right) \mathbf{v}^{(2)}} 
             {2{\mathbf{u}^{(2)}}^*  \mathbf{v}^{(2)}} \right) , 
\end{align}
\end{subequations}
where 
$ \mathbf{u}^*$ denotes the conjugate transpose of vector $ \mathbf{u}$.
\end{theorem}
\begin{proof} 
We first consider $\delta \to -1 $ and note  the identity
\begin{align}
     \frac{1 + \delta}{2}\,& \chi\mathbf{L}^{(1)} + \frac{1 - \delta}{2}\,  (1-\chi)\mathbf{L}^{(2)} \nonumber\\
    = & (1-\chi)\mathbf{L}^{(2)}  + (1 + \delta)\, \frac{\chi\mathbf{L}^{(1)}-(1-\chi)\mathbf{L}^{(2)} }{2}  \nonumber\\
    = &(1-\chi)\mathbf{L}^{(2)}  + (1 + \delta)\, \chi \frac{ \mathbf{L}^{(1)} + \mathbf{L}^{(2)} }{2}  -  (1 + \delta)\, \frac{\mathbf{L}^{(2)} }{2}  \, .
\end{align}
%
We Taylor expand the eigenvalue $\overline{\lambda}_2(\delta)$ and its associated left and right eigenvectors $\overline{{\bf u}}(\delta)$ and $\overline{{\bf v}}(\delta)$, respectively, to find
\begin{subequations}
\begin{align} \label{eq:perturbation_minus1}
    \overline{\lambda}_2(\delta) 
    &= \overline{\lambda}_2(-1) 
    + (1+\delta) \overline{\lambda}^{\prime}_2(  -1) 
    + \mathcal{O}( (1+\delta)^2),  \\
    \overline{{\bf u}}(\delta) &= \overline{ {\bf u}} (-1) + (1+\delta){\overline{ {\bf u}}'}(-1) + \mathcal{O}((1+\delta)^2),  \\
    \overline{ {\bf v}}(\delta) &= \overline{ {\bf v}} (-1) + (1+\delta){\overline{ {\bf v}}'}(-1) + \mathcal{O}((1+\delta)^2). 
\end{align}
\end{subequations}
We substitute the first-order approximations    into the eigenvalue equation $ \overline{\mathbf{L}}(\delta)\overline{ {\bf v}}(\delta)  =  \overline{\lambda}_2(\delta) \overline{ {\bf v}}(\delta) $ to obtain
    \begin{align}
        \Big[ \frac{1 + \delta}{2} & \chi\mathbf{L}^{(1)} + \frac{1 - \delta}{2}  (1-\chi)\mathbf{L}^{(2)}\Big] 
        \Big[ \overline{ {\bf v}} (-1) + (1+\delta){\overline{ {\bf v}}'}(-1)\Big]  \nonumber\\
        =& \Big[ \overline{\lambda}_2(-1) (1+\delta) \overline{\lambda}^{\prime}_2(  -1)  \Big]
         \Big[ \overline{ {\bf v}} (-1) + (1+\delta){\overline{ {\bf v}}'}(-1)\Big] .
    \end{align}

The zeroth-order and first-order terms must both be consistent, which gives rise to two equations. The zeroth-order terms yield an eigenvalue equation
\begin{align}
   (1-\chi)\mathbf{L}^{(2)}  \overline{ {\bf v}} (-1)    =  \overline{\lambda}_2(-1) 
  \overline{ {\bf v}} (-1)   , 
\end{align}
which implies that  $\overline{ {\bf v}} (-1) = {\bf v}^{(2)}$ and $\overline{\lambda}_2(-1)  = (1-\chi) \lambda_2^{(2)}$.
The first-order terms yield a linear equation
    \begin{align}
        (1-\chi)\mathbf{L}^{(2)} &{\overline{ {\bf v}}'}(-1) +    \frac{\chi\mathbf{L}^{(1)}-(1-\chi)\mathbf{L}^{(2)} }{2} {\bf v}^{(2)} \nonumber\\
        = & (1-\chi) \lambda_2^{(2)} {\overline{ {\bf v}}'}(-1) 
        +   \overline{\lambda}^{\prime}_2(  -1)    {\bf v}^{(2)} .
    \end{align}
We left multiply both sides of this equation by $  {\mathbf{u}^{(2)}}^*$ to obtain
\begin{align}
    \chi\, 
    \text{Re} \left( \frac{{\mathbf{u}^{(2)}}^* \left( \mathbf{L}^{(1)} 
    + \mathbf{L}^{(2)} \right) {\bf v}^{(2)}}{2{\mathbf{u}^{(2)}}^*{\bf v}^{(2)} } \right) 
    - \text{Re} \left(  \frac{\lambda_2^{(2)} }{2}  \right)
    =  \overline{\lambda}^{\prime}_2(  -1) .
\end{align}
This complete the analysis for $\delta \to -1$. We repeat this procedure for  $\delta \to 1$ to complete the proof.
\end{proof}

\begin{theorem}
    Consider the following two roots for the linear equations defined in \eqref{eq:lam2_prime44},
    \begin{subequations}
    \begin{align}\label{eq:critical_chi44}
        \hat{\chi}(1) &=  \text{Re} \left( \frac{ {\mathbf{u}^{(1)}}^* \mathbf{L}^{(2)} \mathbf{v}^{(1)}} 
        { {\mathbf{u}^{(1)}}^* (\mathbf{L}^{(1)} + \mathbf{L}^{(2)}) \mathbf{v}^{(1)}} \right) ,\\
        \hat{\chi}(-1) &=  \text{Re} \left( \frac{ {\mathbf{u}^{(2)}}^* \mathbf{L}^{(2)} \mathbf{v}^{(2)}} 
        { {\mathbf{u}^{(2)}}^* (\mathbf{L}^{(1)} + \mathbf{L}^{(2)}) \mathbf{v}^{(2)}} \right),
    \end{align}
    \end{subequations}
    and also define
    \begin{subequations}
    \begin{align}
        s_{1}&= \text{Re} \left( \frac{ {\mathbf{u}^{(1)}}^* \left(\mathbf{L}^{(1)} + \mathbf{L}^{(2)} \right) \mathbf{v}^{(1)} }{ {\mathbf{u}^{(1)}}^* \mathbf{v}^{(1)} } \right), \\
        s_{2} &= \text{Re} \left( \frac{ {\mathbf{u}^{(2)}}^* \left(\mathbf{L}^{(1)} + \mathbf{L}^{(2)} \right) \mathbf{v}^{(2)} }{ {\mathbf{u}^{(2)}}^* \mathbf{v}^{(2)} } \right).
    \end{align}
    \end{subequations}
    Under the assumption that $s_{1}$ and $s_{2}$ are both positive, and $\hat{\chi}(-1), \hat{\chi}(1) \in [0,1]$ with $\hat{\chi}(-1) < \hat{\chi}(1)$, then the convergence rate $\text{Re} \left( \overline{\lambda}_2(\delta) \right)$ for a 2-layer multiplex network is guaranteed to have a cooperative maximum at some value $\delta\in(-1,1)$ if  $\chi \in \Big( \hat{\chi}(-1),  \hat{\chi}(1) \Big) $.  
\end{theorem}

\begin{proof}
    Our proof relies on Rolle's Theorem \cite{sahoo1998mean} for a continuous function $f(\delta)$ on some domain $\delta\in[a,b]$: if $f'(a)>0$ and $f'(b)<0$,  then there exists at least one value of  $\hat{\delta}$ at which the function   $f(\delta) $ obtains its maximum. In our case, $f(\delta)=\text{Re}\left( \overline{\lambda}_2(\delta) \right)$ and $[a,b] = [-1,1]$. 
    Thus, the   maximum is guaranteed to exist  provided that  $\overline{\lambda}_2^{\prime}(-1)>0$ and $\overline{\lambda}_2^{\prime}(1)<0$. 
    
    This criterion can be generally checked for any $\delta$ by simply evaluating $\overline{\lambda}_2^{\prime}(\delta)$ at $\delta =\pm 1$. 
    Moreover, we can  apply these bounds on the right-hand-sides of  \eqref{eq:lam2_prime}   and solve for  $\chi$ to obtain   intervals within which an optimum is guaranteed.
    Depending on the different signs that $s_{1}$ and $s_{2}$ can take and the   values of $\hat{\chi}(-1)$ and $\hat{\chi}(1)$, in principle, different types of intervals are possible. 
In the simplest case, 
%
both $s_{1}$ and $s_{2}$ are positive and the Rolle's Theorem inequalities are equivalent to the inequalities $\chi > \hat{\chi}(-1)$ and $\chi < \hat{\chi}(1)$. Also, assuming that $0 \le \hat{\chi}(-1) < \hat{\chi}(1) \le 1$, then we can conclude that  a cooperative maximum exists when $\chi \in \Big( \hat{\chi}(-1),  \hat{\chi}(1) \Big) $.  
\end{proof}

\section{Extended Study of Optima for   Random Multiplex Networks}\label{sec:appendix_large}

\begin{figure*}[t]
    \centering 
    \includegraphics[width=.99\linewidth]{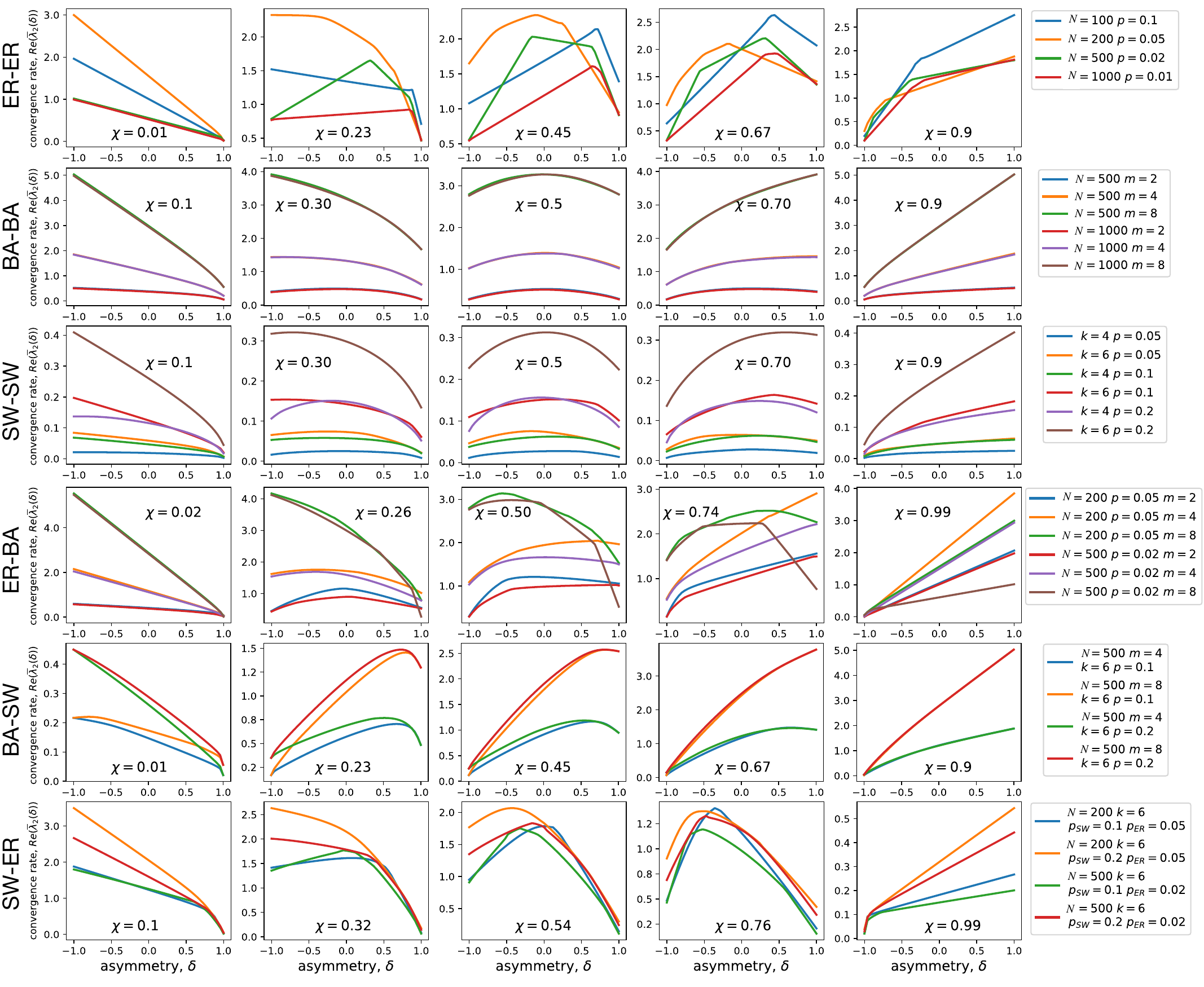}
    \caption{
    {\bf Layers' relative timescales influence whether optima are cooperative vs. non-cooperative for random multiplex networks.}
    We plot  our theoretical prediction given by \eqref{eq:lim_lam2}--\eqref{eq:L_bar_2})  for  the convergence rate 
    versus $\delta$.
    The six rows show $\text{Re}( \overline{\lambda}_2(\delta) )$  for the six random multiplex networks introduced in Sec.~\ref{sec:large}, and the different   curves reflect different parameter choices for the random-graph layers   (see legend).
     The different columns correspond to different choices for the rate-scaling parameter $\chi\in(0,1)$, which  tunes whether layer 1 is much faster ($\chi\approx 1$) or layer 2 is much faster ($\chi\approx 0$). 
     In all cases, we find that the optima are cooperative for intermediate values of $\chi$ (e.g., third column) and non-cooperative when $\chi$ is either too small or large (e.g., left-most and right-most columns). 
    }
    \label{fig:strong_RG}
\end{figure*}

Here, we further study the effects of layers' relative timescales on optima for the family of  random multiplex networks described in Sec.~\ref{sec:large} in which each layer is created using one of three models: the  Erd\"os-R\'enyi (ER) model \cite{erdds1959random},   Barab\'asi-Albert (BA) model \cite{barabasi1999emergence},  or the Watts-Strogatz small-world (SW) model \cite{watts1998collective}. The main extension here is that we now consider many different choices for these generative models.

In Fig.~\ref{fig:strong_RG}, we plot our theoretically predicted convergence rate  $\text{Re}\left( \overline{\lambda}_2(\delta) \right)$ given by \eqref{eq:lim_lam2}--\eqref{eq:L_bar_2}  versus $\delta$. Different columns  reflect different choices for the rate-scaling parameter   $\chi$, and different rows correspond to different generative processes for the random multiplex networks. For example, ``ER-BA''  indicates that layer 1 is created as an ER random graph, while layer 2 is created by the Barab\'asi-Albert (BA)  preferential-attachment model.
%
In each panel, the colored curves depict various parameter choices for the random-graph models (see legends).
By comparing across the columns, observe that their optima are cooperative for intermediate values of $\chi$ (e.g., the third column) and non-cooperative when $\chi$ is either too small or large (e.g., the left-most and right-most columns).

By focusing on the second and fourth columns in Fig.~\ref{fig:strong_RG}, one can observe that some curves yield a cooperative optimum while others do not, depending the parameter choices for the generative models. This allows us to study how network parameters affect $\text{Re}(\lambda_2)$ and whether the optimum is cooperative vs. non-cooperative.
For example, consider ER-BA model in the fourth row, for which we study six parameter-choice combinations. For the first three parameter choices (blue, orange and green curves), we fix the number $N=200$ of nodes and the probability $p=0.05$ for edge creation, and vary  the constant  $m$ used for the BA model. For the last three parameter choices, we consider the same three values of $m$ but decrease $p$ to $p=0.02$. Our first observation is these changes in $m$ appear to have a greater effect than the change to  $p$. Also, observe in   the second column (i.e., $\chi=0.26$) that increasing $m$ can change the optimum from non-cooperative to cooperative for both choices of $p$. Interesting, the opposite can be observed in the   fourth column (i.e., $\chi=0.74$); increasing $m$ changes the optimum from non-cooperative to cooperative. Given this complicated response, we leave open to future research   further investigations into the diverse effects of network parameter choices on $\text{Re}(\lambda_2)$ for random multiplex networks. 



\end{document}